\documentclass[journal]{IEEEtran}

\usepackage{amsmath, amsthm, amssymb, bm, graphicx, hyperref, mathrsfs, color, algorithmicx, cite}
\usepackage[linesnumbered, ruled]{algorithm2e}
\usepackage{hyperref}
\usepackage{subfigure}
\usepackage{changes}
\usepackage{float}
\newtheorem{theorem}{Theorem}%{定理}[section]
\newtheorem{lemma}[theorem]{Lemma}%{引理}
\newtheorem{remark}{Remark}
\begin{document}
%
% paper title
% Titles are generally capitalized except for words such as a, an, and, as,
% at, but, by, for, in, nor, of, on, or, the, to and up, which are usually
% not capitalized unless they are the first or last word of the title.
% Linebreaks \\ can be used within to get better formatting as desired.
% Do not put math or special symbols in the title.
\title{Communication-and-Computation Efficient Split Federated Learning: Gradient Aggregation and Resource Management}%Bare Demo of IEEEtran.cls\\ for IEEE Journals  %%%%%%%%%%%%%  gradient aggregation design and resource %Co-designing Communication, Computation, and Layer Control for management
%
%
% author names and IEEE memberships
% note positions of commas and nonbreaking spaces ( ~ ) LaTeX will not break
% a structure at a ~ so this keeps an author's name from being broken across
% two lines.
% use \thanks{} to gain access to the first footnote area
% a separate \thanks must be used for each paragraph as LaTeX2e's \thanks
% was not built to handle multiple paragraphs
%

\author{Yipeng~Liang,~\IEEEmembership{Graduate Student Member,~IEEE,}
        Qimei~Chen,~\IEEEmembership{Member,~IEEE,}
        Rongpeng~Li,~\IEEEmembership{Member,~IEEE,}
        Guangxu~Zhu,~\IEEEmembership{Member,~IEEE, }
        Muhammad~Kaleem~Awan,~ %\IEEEmembership{Member,~IEEE, }
        and~Hao~Jiang,~\IEEEmembership{Member,~IEEE}%Life~Fellow,~IEEE}% <-this % stops a space
\thanks{Yipeng~Liang, Qimei~Chen, and Hao~Jiang are with the School of Electronic Information, Wuhan University, Wuhan 430072, China (e-mail: liangyipeng@whu.edu.cn, chenqimei@whu.edu.cn, jh@whu.edu.cn).}% <-this % stops a space
\thanks{Guangxu~Zhu is with Shenzhen Research Institute of Big Data, Shenzhen, 518172, China (e-mail: gxzhu@sribd.cn).}
\thanks{Rongpeng~Li is with the College of Information Science and Electronic Engineering, Zhejiang University, Hangzhou 310027, China (e-mail:  lirongpeng@zju.edu.cn).}
\thanks{Muhammad~Kaleem~Awan is with Oman Telecommunications Company Ltd., PC 112, Ruwi, Oman (e-mail: Muhammad.Awan@omantel.om).}
\thanks{This manuscript is a preliminary version of the work and may be subject to further revisions.}
% \thanks{J. Doe and J. Doe are with Anonymous University.}% <-this % stops a space
% \thanks{Manuscript received April 19, 2005; revised August 26, 2015.}
}

\maketitle

% As a general rule, do not put math, special symbols or citations
% in the abstract or keywords.
\begin{abstract}
With the prevalence of Large Learning Models (LLM), Split Federated Learning (SFL), which divides a learning model into server-side and client-side models, has emerged as an appealing technology to deal with the heavy computational burden for network edge clients. 
However, existing SFL frameworks would frequently upload smashed data and download gradients between the server and each client, leading to severe communication overheads. 
To address this issue, this work proposes a novel communication-and-computation efficient SFL framework, which allows dynamic model splitting (server- and client-side model cutting point selection) and broadcasting of aggregated smashed data gradients. 
We theoretically analyze the impact of the cutting point selection on the convergence rate of the proposed framework, revealing that model splitting with a smaller client-side model size leads to a better convergence performance and vise versa.
Based on the above insights, we formulate an optimization problem to minimize the model convergence rate and latency under the consideration of data privacy via a joint Cutting point selection, Communication and Computation resource allocation (CCC) strategy. 
To deal with the proposed mixed integer nonlinear programming optimization problem, we develop an algorithm by integrating the Double Deep Q-learning Network (DDQN) with convex optimization methods.
Extensive experiments validate our theoretical analyses across various datasets, and the numerical results demonstrate the effectiveness and superiority of the proposed communication-efficient SFL  compared with existing schemes, including parallel split learning and traditional SFL mechanisms.

\end{abstract}

% Note that keywords are not normally used for peerreview papers.

% For peer review papers, you can put extra information on the cover
% page as needed:
% \ifCLASSOPTIONpeerreview
% \begin{center} \bfseries EDICS Category: 3-BBND \end{center}
% \fi
%
% For peerreview papers, this IEEEtran command inserts a page break and
% creates the second title. It will be ignored for other modes.
\IEEEpeerreviewmaketitle

\section{Introduction}
% The very first letter is a 2 line initial drop letter followed
% by the rest of the first word in caps.
% 
% form to use if the first word consists of a single letter:
% \IEEEPARstart{A}{demo} file is ....
% 
% form to use if you need the single drop letter followed by
% normal text (unknown if ever used by the IEEE):
% \IEEEPARstart{A}{}demo file is ....
% 
% Some journals put the first two words in caps:
% \IEEEPARstart{T}{his demo} file is ....
% 
% Here we have the typical use of a "T" for an initial drop letter
% and "HIS" in caps to complete the first word.

% \IEEEPARstart{T}{his} demo file is intended to serve as a ``starter file''
% for IEEE journal papers produced under \LaTeX\ using
% IEEEtran.cls version 1.8b and later.

% You must have at least 2 lines in the paragraph with the drop letter
% (should never be an issue)
% I wish you the best of success.

With significant advancements in Artificial Intelligence (AI), future 6G networks are envisioned to transition from ``connected things" to ``connected intelligence" by decentralizing AI from the central cloud to edge networks \cite{LetaiefJSAC2022}. This evolution aims to support edge AI vision in the 6G networks, enabling pervasive intelligence to support emerging intelligent applications such as, eXtended Reality (XR), intelligent transportation systems,  and Internet of Things (IoT) \cite{LetaiefMag2019, SaadNetwork2020, ShiTutorial2020}. 

In this context, Distributed Collaborative Machine Learning (DCML) has emerged as a pivotal technology, particularly due to its inherent data privacy advantages.
Among various DCML approaches, Federated Learning (FL) and Split Learning (SL) have become particularly attractive in recent years. 
FL facilitates the training of a complete Machine Learning (ML) model through collaboration between a central server and distributed clients, without requiring clients to share their local data \cite{LiSPM2020, LiangTWC2024, Chen2022IOTJ}. However, while FL supports parallel model training across multiple clients, resource-constrained devices in edge networks, such as those in IoT environments, often struggle to handle the computational demands of training complete ML models \cite{ChenIOTJ2024}, especially Large Language Models (LLMs). 
To deal with this issue, SL offers a promising solution by offloading a portion of the computational burden to the server. Specifically, SL divides a complete ML model into smaller network portions, deploying one portion on the client and the other on the server. The vanilla SL conducts model training sequentially, with the server interacting with clients one by one to update the model \cite{Vepakomma2018}. However, this sequential approach introduces significant latency, particularly when managing a large number of clients. Moreover, it can lead to catastrophic forgetting, which severely impacts learning performance \cite{DuanSensor2022}. To overcome these limitations, Split Federated Learning (SFL) combines the strengths of FL and SL, enabling parallel model training while alleviating the computational burden on clients \cite{ThapaAAAI2023}. 
With the recent prevalence of LLMs, SFL presents considerable potential for facilitating their training and inference at the edge of 6G networks. %\cite{lin2024}

\begin{figure*}
	\centering  %图片全局居中
    \setlength{\abovecaptionskip}{0.cm}
	\vspace{-0.35cm} %设置与上面正文的距离
	\includegraphics[width=0.8\linewidth]{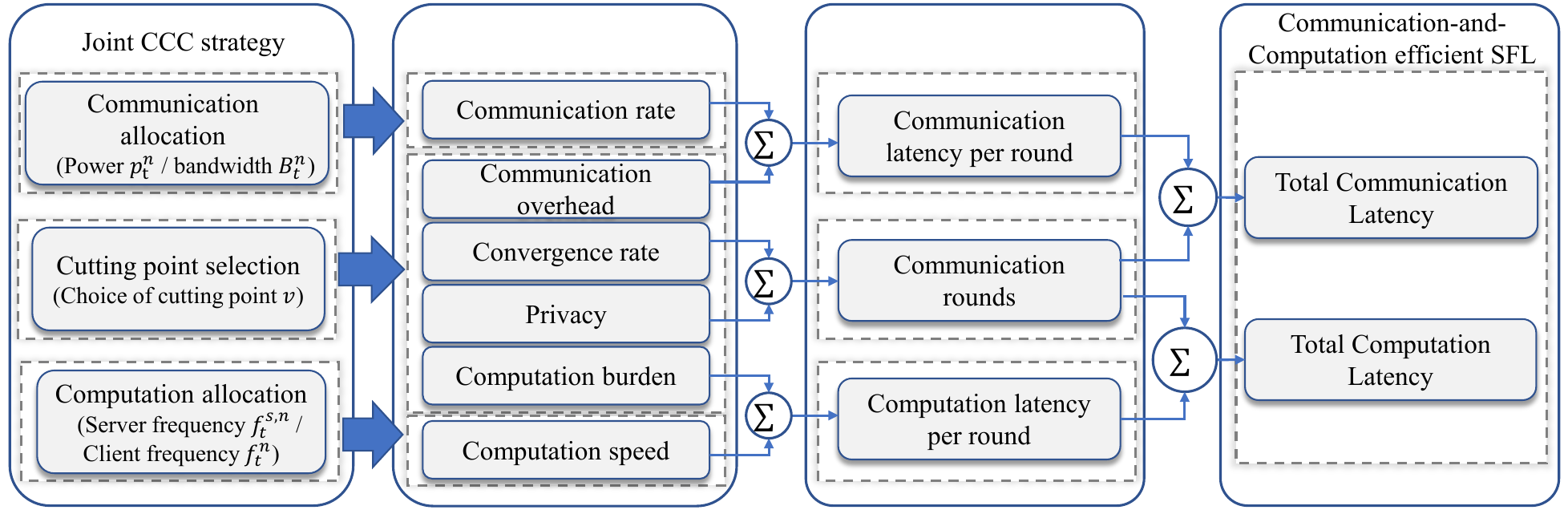}
	\caption{Illustration of the proposed joint CCC strategy.}
	\label{CCCdesign}
	% \vspace{-0.5cm}
\end{figure*}

Despite the advantageous integration of SL and FL, SFL necessitates frequent exchanges of information, such as smashed data and corresponding gradients, between the server and clients to update both server-side and client-side models. Additionally, the synchronous aggregation of client-side models at the server introduces additional communication overhead. Consequently, communication overhead has become a significant challenge in SFL. Although SFL has garnered increasing attention in recent years, efforts to mitigate this communication overhead remain limited. To address this gap, this work proposes a co-design of Cutting point selection, Communication and Computation resource allocation (CCC) strategy for SFL.

\subsection{Related Work}
Different from existing DCML approach such as FL, the research on SFL is still in its early stages. The SFL framework was proposed in \cite{ThapaAAAI2023}, enabling parallel training and synchronous aggregation of both client-side and server-side models, which has recently garnered significant interest in various fields, including medical image segmentation \cite{KafshgariArXiv2023}, wireless networks \cite{XuTWC2023}, and emotion detection \cite{WarefNILES2022}.

Several existing works aim to reduce training latency arising from the sequential training process and client heterogeneity in SFL. 
For example, the authors in \cite{WuJSAC2023} introduced a cluster-based approach that partitions clients into multiple clusters. Within each cluster, client-side models are trained and aggregated in parallel, followed by sequential training of server-side and client-side models across clusters. 
A resource management algorithm was proposed in \cite{WenArxiv2023} to minimize training latency of SFL by jointly optimizing the selection of the cutting point and the allocation of computational resources. 
In \cite{ShenArXiv2023}, an FedPairing scheme was proposed to enhance training efficiency by pairing clients with varying computational resources.
A communication-and-storage efficient SFL framework was explored in \cite{MuICC2023}, where an auxiliary network was integrated into the client-side model to facilitate local updates. This approach maintained a single server-side model at the server, thereby eliminating the need to transmit gradients from the server.

Besides, data privacy and security in SFL have received increasing attention. The authors in \cite{ZhangBigData2023} investigated a privacy-aware SFL, where a client-based privacy approach was introduced to enhance resilience against attacks. In \cite{LeeArXiv2021}, the authors analyzed the tradeoff between privacy and energy consumption in SFL, with a particular focus on the impact of the cutting point selection.

Additionally, some works focus primarily on enhancing SFL performance considering the non-independent and identically distributed (Non-IID) data across clients. For instance, the MergeSFL framework was proposed in \cite{LiaoArXiv2023} addressed Non-IID challenges by employing feature merging and batch size regulation across different clients. 

However, these works suffer from communication overhead due to the synchronous client-side model aggregation. To address this issue, a Parallel Split Learning (PSL) framework was introduced to enhance communication efficiency by eliminating synchronous aggregation.
For instant, the authors in \cite{JeonICOIN2020} presented a PSL method to prevent overfitting through minibatch size selection and client layer synchronization at each client. 
In \cite{LinTMC2024}, a last-layer gradient aggregation scheme was proposed for PSL to reduce training and communication latency at the server. A joint subchannel allocation, power control, and cutting point selection strategy was further proposed, considering heterogeneous channel conditions and computing capabilities among clients.
In \cite{KimWCNC2023}, the authors explored a personalized PSL framework to address Non-IID issues, employing a bisection method with a feasibility test to optimize the tradeoff between energy consumption for computation and wireless transmission, training time, and data privacy.
A local-loss-based training method was proposed in \cite{HanICML2021} to expedite the PSL training process by incorporating an auxiliary network into the client-side model, serving as the local loss function for model updates.

Despite these efforts, previous approaches for SFL still suffer from communication overhead. Specifically, clients are still required to individually upload smashed data to the server and download the corresponding gradients to update both client-side and server-side models. This process leads to substantial communication burden, particularly in environments with limited communication resources.

\subsection{Motivation and Contribution}
Motivated by the above critical issue, we explore a novel communication-and-computation efficient SFL with Gradient Aggregation (SFL-GA) framework in this work. 
Specifically, the SFL-GA framework enables dynamic model cutting point selection based on the wireless communication environment, privacy requirements, and computation abilities of edge devices.  
According to the model splitting, gradients of the smashed data at the server are aggregated, and then broadcasting to all the devices to effectively reduce communication overhead.
For further communication-and-computation efficiency enhancement, we introduce a joint CCC strategy as shown in Fig. \ref{CCCdesign}.  
In detail,  the communication and computation resource allocation schemes influence the communication rate and computation speed, respectively. 
Meanwhile, the model cutting point affects  communication overhead, computation burden, convergence rate, and privacy leakage. 
These elements collectively influence both the communication and computation latency per round, as well as the number of communication rounds required for model convergence. 
Ultimately, the above elements collectively determine the overall communication and computation latency.

Overall, the main contributions of this work are summarized as follows.

\begin{itemize}
    \item \textbf{Communication-and-computation efficient SFL-GA framework:} We propose a novel SFL-GA framework, which enables dynamic model cutting point selection and aggregated gradient broadcasting. 
    Specifically, the cutting point is selected to improve the communication and computation efficiency. 
    On the other hand, we aggregate all the smashed data gradients before broadcasting instead of traditional  individual gradients feedback to each client, thus alleviates the communication overhead.

    \item \textbf{Theoretical convergence analysis and problem formulation:} A theoretical convergence analysis of our proposed framework is conducted, which reveals that cutting a smaller  client-side model leads to better convergence performance. 
    As a cost, it increases the risk of privacy leakage. 
    Based on these insights, we formulate a convergence rate and latency optimization problem under the limitation of communication and computation resources as well as the privacy leakage requirement, which is a Mixed-Integer Non-Linear Programming (MINLP).

    \item \textbf{Joint CCC strategy:} To address the MINLP issues, a joint CCC strategy that integrates double-deep Q-learning (DDQN) algorithm and convex optimization techniques is developed. Specifically, the problem is decomposed into two subproblems: resource allocation and cutting point selection. The resource allocation subproblem is resolved using existing convex optimization methods, while the cutting point selection subproblem is tackled with the DDQN algorithm.

    \item \textbf{Performance evaluation:} Numerical results are conducted to validate the theoretical analyses, and evaluate the superior performance of the proposed SFL-GA mechanism compared with benchmarks, including SFL, PSL, and FL.
    
\end{itemize}

% \hfill mds
 
% \hfill August 26, 2015

% \subsection{Organization}
The rest of this paper is organized as follows. Section \uppercase\expandafter{\romannumeral2} introduces the SFL-GA framework and corresponding system models. In Section \uppercase\expandafter{\romannumeral3}, we theoretically analyze the convergence performance for the proposed framework. In Section \uppercase\expandafter{\romannumeral4}, we formulate the optimization problem and design the resource allocation strategy. Numerical results are presented in Section \uppercase\expandafter{\romannumeral5} followed by a conclusion in Section \uppercase\expandafter{\romannumeral6}.

% \textit{Notations:} 
Throughout the paper, we use the following notation: We use $a$ to denote a scalar, $\mathbf{a}$ is a column vector, $\mathbf{A}$ is a matrix, and $|\cdot| $ represents the modulus operator. The Euclidean norm is written as $\left\| \cdot \right\| $, $\left\langle \mathbf{a}, \mathbf{a}^{\prime} \right\rangle$ is the inner product of $\mathbf{a}$ and $\mathbf{a}^{\prime}$, and $\mathbb{E}$ represents mathematical expectation.

\section{System Model}\label{SysMod}
% \begin{figure}
% 	\centering  %图片全局居中
%     \setlength{\abovecaptionskip}{0.cm}
% 	\vspace{-0.35cm} %设置与上面正文的距离
% 	\includegraphics[width=1\linewidth]{SystemModel.pdf}
% 	\caption{Proposed SFL-GA framework and existing SFL.}
% 	\label{sysmodel}
% 	\vspace{-0.5cm}
% \end{figure}

\begin{figure*}
	\centering
	\subfigure[Traditional Split Federated Learning]{
		\begin{minipage}[b]{0.8 \textwidth}
			% \centering
			\includegraphics[width=\textwidth]{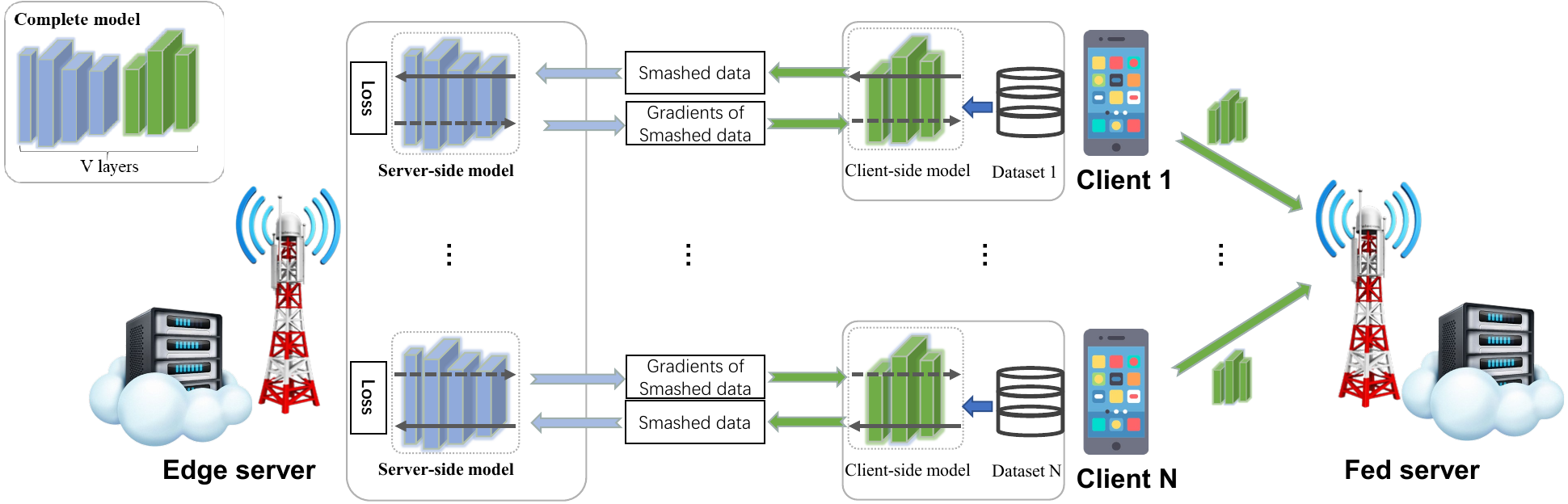} %width=5.0in
		\end{minipage}
	}%
 
	%此处的空行很重要，想让图片在什么地方换行就在代码对应位置空行
	\subfigure[Split Federated Learning with Gradient Aggregation framework]{
		\begin{minipage}[b]{0.8\textwidth}
			% \centering
			\includegraphics[width=\textwidth]{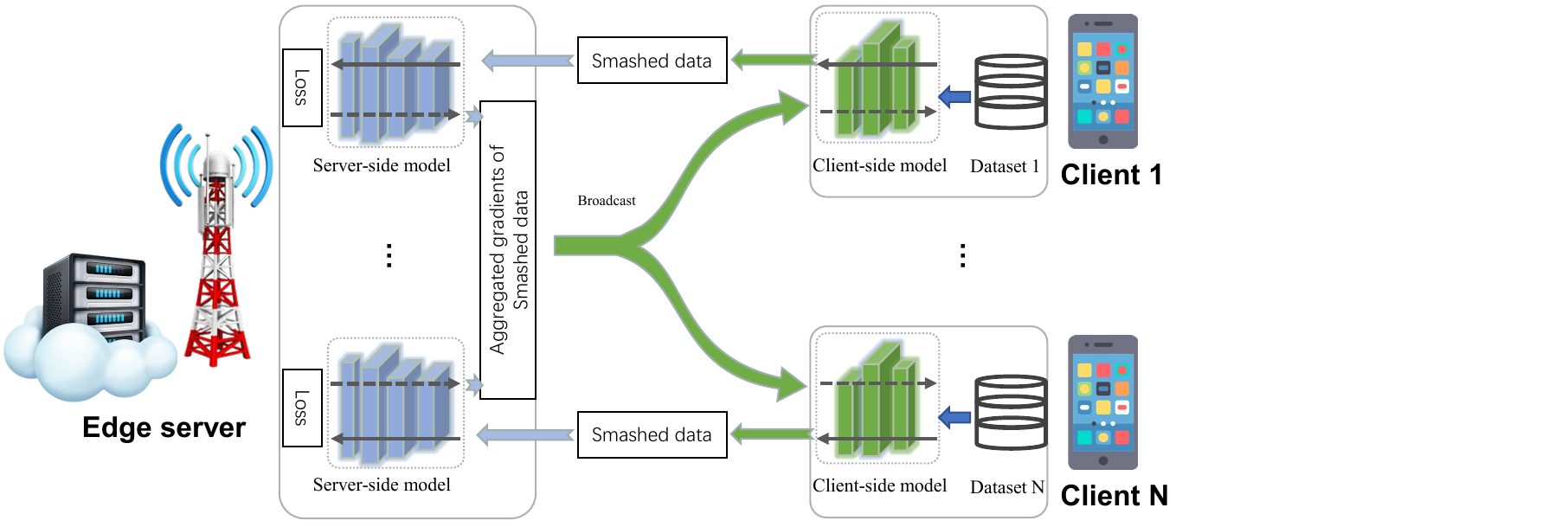} %width=2.2in
		\end{minipage}
	}%
	% \centering
	\caption{Proposed SFL-GA framework and Traditional SFL.}
	\label{sysmodel}
\end{figure*}

In this section, we first introduce a novel SFL-GA framework, and then discuss the related system model.

\subsection{SFL-GA Framework}
As shown in Fig.\ref{sysmodel}, we consider an SFL wireless network with one server and a set of clients denoted by $\mathcal{N} \triangleq \{1, 2, ..., N \}$. All of the clients are collaboratively training a shared Machine Learning (ML) model $\mathbf{w}^{} \in \mathbb{R}^{q}$ of size $q$ and layer $V$ for a specific data analysis task, such as classification and recognition. 
In the SFL framework, each client splits/cuts its learning model at the 
$v \in \mathcal{V}\triangleq\{1, 2,3,...,V-1\}$-th layer. Thus, the entire model divided into the client-side model $\mathbf{w}^{c} \in \mathbb{R}^{\phi\left(v\right)}$ and the server-side model $\mathbf{w}^{s} \in \mathbb{R}^{q-\phi\left(v\right)}$, deployed on the client and server for model training, respectively. Here, $\phi\left(v\right)$ denotes the client-side model size, which depends on the cutting point $v$. 

Therefore, the desired ML model is collaboratively trained by the server and clients over $T$ communication rounds. During each communication round $t \in \mathcal{T} \triangleq \{ 1, 2, ..., T\}$, the training process unfolds as follows:

\begin{itemize}
    \item [1)] \textbf{Smashed data generation}: All clients conduct forward propagation (FP) based on their local datasets and generate smashed data.

    \item [2)] \textbf{Server-side model update}: All clients transmit the smashed data, along with the corresponding labels, to the server. Thereafter, the server performs FP and calculates the loss function in parallel based on the received smashed data and their corresponding labels. 
    Then, the server executes back propagation (BP) to obtain the updated server-side models and the smashed data gradients.

    \item [3)] \textbf{Server-side  gradients aggregation}: The server aggregates these updated server-side models into a new global server-side model, as well as aggregates the gradients of the smashed data.

    \item [4)] \textbf{Gradient  broadcast}: The server %aggregates the gradients of the smashed data and then
    broadcasts the aggregated gradients of the smashed data to all the clients. Unlike the traditional SFL, where the gradients of smashed data are transmitted to corresponding clients separately, our proposed framework effectively reduces communication overhead.

    \item [5)] \textbf{Client-side model BP}: The clients perform client-side model BP and update their client-side models based on the received gradients. Since the client-side models have identical weight parameters and are updated based on the same gradients, synchronous client-side model aggregation is eliminated, leading to a further communication overhead reduction.

\end{itemize}

The core concept of the proposed framework is to aggregate the gradients of smashed data from all the clients and then broadcast the aggregated gradients to these clients. Compared with the traditional SFL mechanism, it can significantly reduced the communication overhead.

\subsection{SFL-GA Model}
Denote that the complete ML model $\mathbf{w}^{}$ is partitioned into the server-side model $\mathbf{w}^{s}$ and the client-side model $\mathbf{w}^{c}$, denoted as $\mathbf{w}^{} = \left[\mathbf{w}^{s}; \mathbf{w}^{c}\right]$. %, as shown in Fig.1
Moreover, each client $n$ is assumed to possess a local dataset $\mathcal{D}^n$ with size of $D^n$. 

For communication round $t$, all clients perform client-side model FP in parallel based on their own client-side models $\mathbf{w}^{c}_{t-1}$. Specifically, each client $n$ takes a mini-batch of samples $\xi^n$, randomly chosen from $\mathcal{D}^n$, as the input for $\mathbf{w}^{c}_{t-1}$ to obtain the smashed data, which can be expressed as
\begin{equation}
     \mathbf{S}_{t}^{n} = \ell \left(\mathbf{w}^{c}_{t-1}; \xi^n \right),  %\xi^n
 \end{equation}
where $\ell$ is the the client-side function mapping from the input data to the smashed data. %$\xi^n \in \mathcal{D}^n$ is a mini-batch of samples randomly chosen form $\mathcal{D}^n$.
Then, the server collects the smashed data along with their corresponding labels from all the clients for the server-side model update and aggregation. Specifically, the server first performs  FP to calculate the loss for each device by inputting the corresponding smashed data into the server-side model $\mathbf{w}^{s}_{t-1}$. Therefore, the local loss with respect to the complete model of device $n$ can be expressed as
\begin{equation}\label{localloss}
    \begin{split}
     F\left( \mathbf{w}^{ n}_{t-1} \right) &= F\left( \mathbf{w}^{s, n}_{t-1}, \mathbf{w}^{c}_{t-1}; \xi^n \right) \\
     &= \frac{1}{D^n} \sum_{(x_j, y_j) \in \xi^n } f \left( \mathbf{w}^{s, n}_{t-1}, \mathbf{w}^{c}_{t-1}; (x_j, y_j) \right),
    \end{split}
    \end{equation}
where $f$ is the loss function. $(x_j, y_j)$ is the $j$-th sample of $\xi^n$  with data $x_j$ and label $y_j$. % 
Subsequently, after server-side model BP, the server obtains the gradients of the loss function for client $n$'s model update, which is given by
\begin{equation}
    \mathbf{g}^{s,n}_{t-1} = \nabla_{\mathbf{w}^{s}} F\left( \mathbf{w}^{s, n}_{t-1}, \mathbf{w}^{c}_{t-1}; \xi^n  \right).
\end{equation}
Meanwhile, the gradients of the smashed data of each device $n$ are also computed as 
\begin{equation}
   \mathbf{s}_t^n = \nabla \mathbf{S}_{t}^{n}.
\end{equation} 

Different from the existing SFL in \cite{ThapaAAAI2023} where the server sends $\mathbf{s}_t^n$ to the corresponding client $n$, the server in this work broadcasts the aggregated gradients of smashed data to all clients, which is given by
\begin{equation} \label{aggreGradient}
    \mathbf{s}_t = \sum_{n=1}^{N} \rho^n \mathbf{s}_t^n,
\end{equation}
where $\rho^n = \frac{D^n }{D}$ with $ D = \sum_{n=1}^{N}D^n $ being the total size of datasets across clients.

After receiving $\mathbf{s}_t$ from the server, both the server and the clients can update the server- and client-side models, respectively.
Specifically, each client first computes the gradient of its client-side model, denoted by  $\mathbf{g}^{c}_{t}$, based on $\mathbf{s}_t$. The client-side model is then updated via gradient descent. Consequently, the complete ML model $\mathbf{w}^{n}_{t}$ of the client $n$ in the $t$-th communication round can be updated through multiple epochs, described as follows
\begin{equation}\label{modelupdate}
 \mathbf{w}^{n}_{t} =   
 \left[ 
\begin{matrix}
 \mathbf{w}^{s,n}_{t}\\
 \mathbf{w}^{c}_{t}
\end{matrix}
\right] 
= \left[ 
\begin{matrix}
 \mathbf{w}^{s}_{t-1}\\
 \mathbf{w}^{c}_{t-1}
\end{matrix}
\right]  
-\eta \sum_{i = 1}^{\tau} \left[ 
\begin{matrix}
 \mathbf{g}^{s,n}_{t-1,i}\\
 \mathbf{g}^{c}_{t-1,i}
\end{matrix}
\right]  ,
\end{equation}
where $\eta$ is the learning rate. $\tau$ is the number of local epochs. Since the model $\mathbf{w}^{c}_{t}$ at each client  is updated based on the same gradient $\mathbf{g}^{c}_{t-1}$ and the same parameters $\mathbf{w}^{c}_{t-1}$, each client attains the same model parameters after updating. Therefore, the proposed SFL-GA eliminates the necessity for client-side model aggregation as  \cite{ThapaAAAI2023}. Consequently, we only need to aggregate the server-side model at the server, which is given by
\begin{equation}
    \mathbf{w}^{s}_{t} = \sum_{n=1}^{N} \rho^n \mathbf{w}^{s,n}_{t}.
\end{equation}

As a result, with the complete global model $\mathbf{w}^{}_{t} = \left[ \mathbf{w}^{s}_{t}, \mathbf{w}^{c}_{t} \right] $ at the $t$-th communication round, the global loss function can be represented as
\begin{equation}\label{loss}
     F\left( \mathbf{w}^{}_{t}\right) = \sum_{n=1}^{N} \rho^n F\left( \mathbf{w}^{n}_{t}\right).
\end{equation}

Assuming the global model converges after $T$ communication rounds, the training objective of Eq. \eqref{loss} is to find a minimal global model $\mathbf{w}^* =\left[\mathbf{w}^{s*}; \mathbf{w}^{c*}\right]$ that satisfies
\begin{equation}
    \mathbf{w}^* = \arg \min_{\mathbf{w}} F\left( \mathbf{w}^{}_{T}\right) .
\end{equation}

\subsection{Communication Model}
In this subsection, we introduce the communication model in this work. 
We assume that the channel remains constant during a given communication round but may vary across different rounds. 
For an arbitrary communication round $t$, the communication process of each client includes an uplink phase for uploading smashed data and a downlink phase for broadcasting gradients.

Regarding the uplink communication, the total available bandwidth $B$ is divided into multiple orthogonal subchannels to transmit the smashed data and labels from each client to the server. Therefore, the achievable data rate of device $n$ can be expressed as %\mathrm{U}
\begin{equation}
    r_t^{n,U}=B^n_t \log _2\left(1+\frac{p^n_t g^n_t}{B^n_t N_0}\right) ,
\end{equation}
where $B^n_t$ is the bandwidth allocated to client $n$. $g^n_t$ is the channel gain between the server and device $n$. $p^n_t$ and $N_0$ denote the transmit power of device $n$ and the thermal noise spectrum density, respectively.

For the downlink communication, the server broadcasts the aggregated gradients \eqref{aggreGradient}, occupying the entire bandwidth. Thus, the achievable rate for device $n$ can be represented as 
\begin{equation}
    r_t^{n,D}=B \log _2\left(1+\frac{P g^n_t}{B N_0}\right),
\end{equation}
where $P$ denotes the transmit power of server.

Note that the size of smashed data and corresponding gradients depend on the cutting point $v$. Therefore, the latency of uplink and downlink transmission can be respectively written as %, and remain unchanged during every communication round
\begin{equation}
    l^{n,U}_{t} = \frac{X_{t}\left( v\right)}{r_t^{n,U}},
\end{equation}
\begin{equation}
    l^{n,D}_{t} = \frac{X_{t}\left( v\right)}{r_t^{n,D}},
\end{equation}
where $X_{t}\left( v\right)$ is the communication bit size of smashed data (and its gradient) related to cutting point $v$.%, and $X_1$ is the communication bit size of labels.

\subsection{Computation Model}
In this subsection, we present the computation model of the proposed SFL-GA. During communication round $t$, SFL initiates with client-side model FP.
Let $\gamma^{n}_F\left(v\right)$ denote the computation workload (in FLOPs) of client $n$ for performing FP with one data sample \cite{ ZhangCVPR2018, ZengTWC2021}.
Therefore, the latency for client-side model FP is given by
\begin{equation}
    l^{n,F}_{t} = \frac{D^n\gamma^{n}_F\left(v\right)}{f^n_t},  %\kappa
\end{equation}
where $f^n_t$ denotes the central processing unit (CPU) resource of device $n$. % and $\kappa$ denotes the computing intensity \cite{WuJSAC2023}.

Subsequently, SFL performs FP and BP to update the server-side model at the server. Let $\gamma^{s}_F\left(v\right)$ and $\gamma^{s}_B\left(v\right)$ denote the computation workload of the server for performing FP and BP with one data sample, respectively. Therefore, the latency for both server-side model FP and BP is given by
\begin{equation}
    l^{n,s}_{t} = \frac{D^n\left( \gamma^{s}_F\left(v\right)  + \gamma^{s}_B\left(v\right) \right)} { f^{s,n}_{t}},
\end{equation}
where $f^{s,n}_{t}$ is the CPU computation resource of server that allocated to server-side model of client $n$.

Finally, SFL-GA conducts client-side model BP to update the model at each client. Let $\gamma^{n}_B\left(v\right)$ denote the computation workload of client $n$ for performing BP with one data sample. Then, we have
\begin{equation}
    l^{n,B}_{t} = \frac{D^n\gamma^{n}_B\left(v\right)}{f^n_t}. %\kappa
\end{equation}

\subsection{Privacy  Model}
In the proposed SFL framework, the privacy concerns arise from the transmission of smashed data between clients and the server. 
Existing research has demonstrated the significant impact of cutting point on privacy leakage\cite{FredriksonACMCCCS2015, LeeArXiv2021, KimSMA2020}. Generally, a greater number of layers in the client-side model leads to poorer reconstruction results, making it more difficult to infer the original input data. Consequently, as the number of layers increases, the potential for privacy leakage diminishes. This insight underscores the importance of carefully selecting the cutting point in SFL to balance computational efficiency with privacy preservation.

In this work, to quantify privacy leakage, we adopt the privacy model introduced in \cite{KimWCNC2023}, which is formulated as
\begin{equation}
    \log(1 + \frac{\phi^{}_{t}\left(v\right)}{q}) \geq \epsilon,  \forall t,
\end{equation}
where $\epsilon$ represents a desired threshold for privacy protection.

\section{Theoretical Analysis and Performance Evaluation}
In this section, we present a theoretical analysis of the proposed SFL-GA framework. We start by analyzing the convergence of SFL-GA, followed by a discussion of its complexity and scalability.

\subsection{Convergence Analysis}
To facilitate analysis, we denote the mini-batch (stochastic) gradient and full-batch gradient of the loss function, respectively, as
\begin{equation}
     \mathbf{g}^{n}_{t} = \left[ \mathbf{g}^{s,n}_{t} ;  \mathbf{g}^{c}_{t} \right],       %\nabla_{\mathbf{w}^{s}} F\left( \mathbf{w}^{n}_{t}\right)
\end{equation}
and
\begin{equation}
    \mathbf{h}^{n}_{t} = \left[\nabla_{\mathbf{w}^{s}} F\left( \mathbf{w}^{n}_{t}\right); \tilde{\nabla_{\mathbf{w}^{c}} F\left( \mathbf{w}^{}_{t}\right)} \right].
\end{equation}

Recall that in traditional SFL, each client $n$ generates the gradient of the client-side model $\nabla_{\mathbf{w}^{c}} F\left( \mathbf{w}^{n}_{t}\right)$ based on its own gradient of smashed data $\mathbf{s}^n_t$, instead of the aggregated one in Eq. \eqref{modelupdate}, to update $\mathbf{w}^{c}_{t}$. This discrepancy may affect the convergence performance.
Therefore, we denote the full-batch gradients as 
\begin{equation}
    \nabla F(\mathbf{w}^{n}_{t}) = \left[\nabla_{\mathbf{w}^{s}} F\left( \mathbf{w}^{n}_{t}\right); \nabla_{\mathbf{w}^{c}} F\left( \mathbf{w}^{n}_{t}\right) \right],
\end{equation}
Accordingly, we further define the the SFL's global gradient as
\begin{equation}
    \begin{split}
        \nabla F(\mathbf{w}^{}_{t}) & = \sum_{n=1}^{N}\nabla F(\mathbf{w}^{n}_{t}) \\
          & = \left[\sum_{n=1}^{N} \nabla_{\mathbf{w}^{s}} F\left( \mathbf{w}^{n}_{t}\right); \sum_{n=1}^{N} \nabla_{\mathbf{w}^{c}} F\left( \mathbf{w}^{n}_{t}\right) \right].
    \end{split}
\end{equation}

To begin with, we introduce the following assumptions, which are commonly adopted in existing works, such as \cite{XuTWC2023,  HanICML2021}.%LiICLR2023,

\textbf{Assumption 1 (L-smoothness).} For any $\mathbf{w}, \mathbf{v}$, the loss function  is either continuously differentiable or Lipschitz continuous with a non-negative Lipschitz constant $L\geq 0$, which can be formulated as % \boldsymbol{v}
\begin{equation}
    F(\mathbf{v}) - F(\mathbf{w}) \leq (\mathbf{v}-\mathbf{w})^{\top} \nabla F(\mathbf{w})+\frac{L}{2}\|\mathbf{v}-\mathbf{w}\|^2.
\end{equation}

\textbf{Assumption 2 (Unbiased Gradient and Bounded Variance ).} 
% Let $\xi^n$ be a mini-batch of samples randomly chosen from the $n$-th client’s local data uniformly.  
For each client, the stochastic gradient is unbiased, i.e., $\mathbb{E} \left(\mathbf{g}^{s,n}_{t} \right) = \nabla_{\mathbf{w}^{s}} F\left( \mathbf{w}^{n}_{t}\right) $ and $\mathbb{E} \left( \mathbf{g}^{c}_{t} \right) = \tilde{\nabla_{\mathbf{w}^{c}} F\left( \mathbf{w}^{n}_{t}\right)} $. Moreover, 
the variance of stochastic gradients of each client is bounded by
\begin{equation}
    \begin{split}
        \mathbb{E} \left(  \|\mathbf{g}^{n}_{t} -  \nabla F (\mathbf{w}^{n}_{t} )\|^2 \right) &  \leq \sigma^2. 
    \end{split}
\end{equation}

Note that in both the SFL-GA framework and the traditional SFL framework, the server-side model is updated using the same smashed data. The divergence between the two frameworks arises only in the updating of the client-side model: SFL-GA employs aggregated gradients of the smashed data, while the traditional SFL framework uses the gradients from each client's individual smashed data. This divergence influences the convergence behavior, and accurately characterizing this discrepancy is challenging. Nevertheless, we can observe that this discrepancy is significantly related to the size of the client-side model. Specifically, as the size of the client-side model increases, the differences between the client-side models in the SFL-GA framework and the traditional SFL framework become more pronounced, thereby exerting a greater impact on the convergence of both frameworks. Considering the convergence discrepancy related to the size of the client-side model, we introduce the following assumption.

\textbf{Assumption 4. (Bounding the Difference Between SFL Gradient and SFL-AG Gradient).} 
When the client-side model holds a size of $\phi_t\left( v\right)$ in $t$-th round, the expected gradient variance between SFL and SFLAG is bounded by 
\begin{equation}
     \begin{split}
         \mathbb{E} \left(  \left\|\mathbf{h}^{n}_{t} -  \nabla F(\mathbf{w}^{n}_{t}) \right\|^2 \right) & = \mathbb{E} \left(  \left\|   \nabla_{\mathbf{w}^{s}} F\left( \mathbf{w}^{n}_{t}\right) -  \nabla_{\mathbf{w}^{s}} F\left( \mathbf{w}^{n}_{t}\right) \right\|^2 \right. \\
		& \left. ~~~ +  \left\| \tilde{\nabla_{\mathbf{w}^{s}} F\left( \mathbf{w}^{}_{t}\right)} -  \nabla_{\mathbf{w}^{c}} F\left( \mathbf{w}^{n}_{t}\right) \right\|^2 \right)\\
         &= \mathbb{E} \left\| \tilde{\nabla_{\mathbf{w}^{c}} F\left( \mathbf{w}^{}_{t}\right)} -  \nabla_{\mathbf{w}^{c}} F\left( \mathbf{w}^{n}_{t}\right) \right\|^2 \\
         &\leq \Gamma\left(\phi^{}_{t}\left(v\right)\right),
     \end{split}
\end{equation}
where $\Gamma\left( \cdot \right)$ is a monotone non-decreasing function with respect to the client-side model size $\phi^{}_{t}\left(v\right)$.
Assumption 4 indicates that the gradient difference between vanilla SFL and SFL-GA is related to the client-side model size $\phi_t\left( v\right)$ and is bounded. Specifically, a smaller client-side model size $\phi_t\left( v\right)$ results in a smaller gradient difference.

With above assumptions, we introduce Lemma \ref{Lemma1} to demonstrate the upper bound of the improvement of the global loss function in each round. 

\begin{lemma} \label{Lemma1}
    When  the learning rate $\eta$ satisfies $0 \leq 2 L^2\eta^2 \tau \left( \tau -1 \right)  \leq \frac{1}{5}$ in the $t$-th communication round, the improvement of the global loss function is bounded by
    \begin{equation}
    \begin{split}
        \mathbb{E} \left( F\left(\mathbf{w}_{t+1}\right) - F\left(\mathbf{w}_{t}\right)\right) \leq 
        - \frac{\eta\tau}{4} \left\| \nabla F(\mathbf{w}^{}_{t}) \right\|^2 + \eta\tau \Gamma\left(\phi^{}_{t}\left(v\right)\right)\\
        + L\eta^2 \tau \sigma^2 \sum_{n=1}^{N}\left(\rho^{n}\right)^2  + \frac{5 L^2\eta^3 \sigma^2 \tau_{} \left( \tau_{} -1\right)}{4}. 
    \end{split}
\end{equation}

\end{lemma}

\begin{proof}
    Please refer to Appendix \ref{AppenB}.
\end{proof}

Based on Lemma \ref{Lemma1}, we further introduce the following Theorem to show the upper bound
of the average squared gradient norm, which illustrates the convergence performance.
\begin{theorem}
    Under the condition of $0 \leq 2 L^2\eta^2 \tau \left( \tau -1 \right)  \leq \frac{1}{5}$,  the average squared gradient norm after T communication rounds is bounded by
\begin{equation}\label{ASGN}
    \begin{split}
    \frac{1}{T} \sum_{t=1}^{T}\left\| \nabla F(\mathbf{w}^{}_{t}) \right\|^2 \leq \underbrace{\frac{4\left( F\left(\mathbf{w}_{T}\right) - F^*\right)}{\eta\tau T}}_{\text{Effects of initialization}} + \underbrace{ \frac{4 }{T} \sum_{t=1}^{T} \Gamma\left(\phi^{}_{t}\left(v\right)\right)}_{\text{Effects of cutting point}} \\
    + \underbrace{4 L\eta \sigma^2 \sum_{n=1}^{N}\left(\rho^{n}\right)^2 + 5L^2\eta^2 \sigma^2\left( \tau_{} -1\right)}_{\text{ Effects of gradient variance} }. 
    \end{split}
\end{equation}
\end{theorem}

\begin{remark}
    It is observed that the \eqref{ASGN} is influenced by the initialization, gradient variance, as well as the cutting point. In particular, reducing the impact of the cutting point during the training process can lower the bound of \eqref{ASGN}, suggesting that a smaller client-side model size enhances convergence performance. At the same time, the cutting point also affects communication overhead, computational burden, and privacy leakage. 
    This observation motivates the co-design of CCC for achieving a communication-and-computation efficient SFL.
\end{remark}

\subsection{Computational Complexity and Scalability}
To facilitate the analysis of computational complexity, we let $\rho^{n}_{} = \frac{1}{N}$. Therefore, \eqref{ASGN} can be reformulated as
\begin{equation}
    \begin{split}
    \frac{1}{T} \sum_{t=1}^{T}\left\| \nabla F(\mathbf{w}^{}_{t}) \right\|^2 \leq \frac{4\left( F\left(\mathbf{w}_{T}\right) - F^*\right)}{\eta\tau T} + \frac{4 L\eta \sigma^2}{N} \\
    + 5L^2\eta^2 \sigma^2\left( \tau_{} -1\right)+  \frac{4 }{T} \sum_{t=1}^{T} \Gamma\left(\phi^{}_{t}\left(v\right)\right). 
    \end{split}
\end{equation}

If the learning rate satisfies $\eta = \sqrt{\frac{N}{\tau T}}$ \cite{LiangTWC2024}, the computational complexity of SFL-GA is given by
\begin{equation}\label{complexitySGLGA}
\begin{split}
    & \mathbb{E}\left(\frac{1}{T} \sum_{t=1}^{T}\left\| \nabla F(\mathbf{w}^{}_{t}) \right\|^2 \right) \\
    & \leq \mathcal{O}\left(\frac{ 1}{\sqrt{\tau N T}}  +  \frac{\sigma^2}{\sqrt{\tau NT}} + \frac{N\left( \tau_{} -1\right)\sigma^2}{\tau T} + M \right),
\end{split}
\end{equation}
where $ M = \frac{4 }{T} \sum_{t=1}^{T} \Gamma\left(\phi^{}_{t}\left(v\right)\right) $. If $M$ is bounded, the computational complexities of proposed SGL-GA is given by $\mathcal{O} \left( \frac{1 }{\sqrt{\tau N T}}  \right) + \mathcal{O} \left(  \frac{\sigma^2 N  \left( \tau_{} -1\right) }{\tau_{} T } \right)$. % which is consistent with previous results \cite{}

To evaluate the scalability of our proposed SFL-GA, we focus on the convergence behavior in \eqref{complexitySGLGA} with respect to the number of clients $N$. 
From \eqref{complexitySGLGA}, the first two terms decrease with the increment of $N$. It indicates that the convergence rate benefits from a larger number of clients through improving the average updates performance. 
On the other aspect, the third term of \eqref{complexitySGLGA} increases linearly with $N$, which suggests a potential risk in computational complexity and variance. Consequently, the convergence behavior initially improves but eventually deteriorates as $N$ continues to grow. 

\section{Problem Formulation and Resource Optimization}
In this section, we formulate a convergence rate and latency optimization problem based on the system models and convergence results. Subsequently, a joint CCC strategy that solved by the DDQN algorithm and convex optimization  is designed.

\subsection{Problem Formulation}
According to the latency analysis in Section \ref{SysMod}, the gradients of smashed data are aggregated before broadcasting to all clients. Therefore, the total latency for communication round $t$ is derived by
\begin{equation}
    l_t = \max_{n} \{ l^{n,U}_{t} + l^{n,F}_{t} + l^{n,s}_{t}  \} + \max_{n} \{ l^{n,D}_{t} +  l^{n,B}_{t} \}.
\end{equation}

Apparently, the wireless channel conditions and heterogeneous computing capabilities of clients may significantly affect the latency of SFL, while the cutting point influences the convergence rate. In this work, we aim to realize communication-and-computation efficient SFL. To this end, we formulate a convergence rate and latency optimization problem while considering resource and privacy constraints.

Let $\boldsymbol{f}^{s} = \left[f^{1, s}_{1}, ..., f^{N, s}_{T} \right]$, $\boldsymbol{f}^{c} = \left[f^{1, c}_{1}, ..., f^{N, c}_{T} \right]$, $\boldsymbol{p} = \left[p^{1}_{1}, ..., p^{N}_{T} \right]$ and $\boldsymbol{B} = \left[B^{1}_{1}, ..., B^{N}_{T} \right]$. The communication-and-computation efficient problem can be formulated as
\begin{align}
\mathcal{P} 1: \mathop {\min }\limits_{ \left\{ v, \boldsymbol{f}^{s}, \boldsymbol{f}^{c}, \boldsymbol{p}, \boldsymbol{B}\right\}}~ & \sum_{t=1}^{T} \left( w \Gamma\left(\phi^{}_{t}\left(v\right)\right) +  l_t \left( v\right)\right),   \label{1} \\ %\sum_{n=1}^{N}  \rho^{n} \delta^{n}_{t}\left( v\right)
\mathrm{s.t.} \qquad & v \in \mathcal{V}, \tag{\ref{1}{a}} \label{1a}\\ %\in \mathcal{T}
& 0 \leq f^{n}_{t} \leq f^{n,c}_{\max},  \ \forall n,t, \tag{\ref{1}{b}} \label{1b}\\
& 0 \leq p^{n}_{t} \leq  p^{n}_{\max},  \forall n,t, \tag{\ref{1}{c}} \label{1c}\\
& \sum_{n=1}^{N} f^{n,s}_{t} \leq f^{s}_{\max}, \  \forall t,  \tag{\ref{1}{d}} \label{1d}\\
& \log(1 + \frac{\phi^{}_{t}\left(v\right)}{q}) \geq \epsilon,  \forall t, \tag{\ref{1}{e}} \label{1e}\\
& \sum^{N}_{n=1} B^{n}_{t} \leq B, \ \forall t ,\tag{\ref{1}{f}} \label{1f}
\end{align}
where $ w$ is a weighted factor to balance the convergence rate and latency. \eqref{1a} is the layer constraint of ML model, $f^{n,c}_{\max}$ and $p^{s}_{\max}$ in \eqref{1b} and \eqref{1c} are the the maximum computation resource and transmit power of each client, respectively, $f^{n,s}_{\max}$ in \eqref{1d} denotes the maximum computation resource constraint for model update at the server. \eqref{1e} is the constraint for privacy protection, \eqref{1f} ensures that the total bandwidth for all clients doesn't exceed the available bandwidth $B$.

$\mathcal{P}1$ is a min-max MINLP problem, which is coupling of cutting point selection and resource allocation. The major difficulty of solving $\mathcal{P}1$ lies in the efficient determination of the cutting point selection under dynamic fading channels and heterogeneous capabilities of clients. Existing optimization algorithms require iteratively adjusting or exhaustively enumerate the cutting point variables, which is inefficient. Therefore, a joint CCC strategy is developed based on DDQN algorithm.

\subsection{Joint CCC  Strategy}  

To address the min-max issue of $\mathcal{P}1$, we introduce auxiliary variables $\chi_t$ and $\psi_t$. Then, $\mathcal{P}1$ can be equivalently reformulated as
\begin{align}
\mathcal{P} 2: \mathop {\min }\limits_{ \left\{ \boldsymbol{v}, \boldsymbol{f}^{s}, \boldsymbol{f}^{c}, \boldsymbol{p}, \boldsymbol{B}\right\}}~ & \sum_{t=1}^{T} \left( w \Gamma\left(\phi^{}_{t}\left(v\right)\right) + \chi_t +  \psi_t\right),   \label{2} \\ 
\mathrm{s.t.} \qquad & \eqref{1a}, \eqref{1b}, \eqref{1c}, \eqref{1d}, \eqref{1e}, \eqref{1f} \tag{\ref{2}{a}}, \label{2a}\\ 
&l^{n,U}_{t} + l^{n,F}_{t} + l^{n,s}_{t}  \leq \chi_t  \tag{\ref{2}{b}}, \label{2b}\\
&l^{n,D}_{t} +  l^{n,B}_{t} \leq \psi_t \ \forall t . \tag{\ref{2}{c}} \label{2c}
\end{align}
Nonetheless, $\mathcal{P} 2$ remains an NP-hard problem. Intuitively, it can be divided into two subproblems: resource allocation and cutting point selection. Therefore, we can address $\mathcal{P} 2$ by jointly solving these subproblems. 
In what follows, we present the optimization methods for the resource allocation and cutting point selection subproblems, respectively.

\subsubsection{Resource Allocation} 
Given the optimal cutting point selection variables $\boldsymbol{v}^{*}$, the resource allocation subproblem is independent to the communication rounds.  Therefore, it can be decomposed into $T$ separate subproblems, each addressed independently.  Without loss of generality, the resource allocation subproblem for communication roud $t$ is formulated as
\begin{equation}\label{2.1}
    \mathcal{P} 2.1: \mathop {\min }\limits_{ \left\{  \boldsymbol{f}^{s}, \boldsymbol{f}^{c}, \boldsymbol{p}, \boldsymbol{B}\right\}} \chi_t +  \psi_t,   
\end{equation}
subject to \eqref{1b}, \eqref{1c}, \eqref{1d},  \eqref{1f}, \eqref{2b}, \eqref{2c}.

It can be easily shown that $\mathcal{P} 2.1$ is a convex optimization problem. Therefore, it can be resolved by existing optimization technique (e.g., CVX).

\subsubsection{Cutting Point Selection}
Given the optimal resource allocation $\boldsymbol{f}^{*}$, $\boldsymbol{p}^{*}$ and $\boldsymbol{B}^{*}$. The cutting point selection subproblem can be expressed as
\begin{equation}
    \mathcal{P} 2.2: \mathop {\min }\limits_{ \left\{ \boldsymbol{v}\right\}} \sum_{t=1}^{T} \left( w \Gamma\left(\phi^{}_{t}\left(v\right)\right) + \chi_t +  \psi_t\right), 
\end{equation}
subject to \eqref{1a}, \eqref{1e},\eqref{2b}, \eqref{2c}.

Due to the integer variables, $\mathcal{P} 2.2$ is an integer programming.
The DDQN algorithm has been identified as an effective method to tackle such decision issues with integer variables \cite{ZhangJSAC2022, PeiJSAC2020, HuangTMC2020}. Therefore, we employ the DDQN algorithm to deal with $\mathcal{P} 2.2$ in this work. 
Before using the DDQN algorithm, the subproblem $\mathcal{P} 2.2$ needs to be transformed into a Markov Decision Process (MDP) problem with a tuple $\left\langle \mathcal{S}, \mathcal{A},  \mathcal{P}, \mathcal{R} \right\rangle$, where $\mathcal{S}$,   $\mathcal{A}$,  $\mathcal{P}$ and $\mathcal{R}$ are the state space, action space, state transition probability, and reward, respectively. Specifically, the corresponding elements in the tuple are presented as follows.
\begin{itemize}
    \item  State space $\mathcal{S}$. It is observed the channel gain at the beginning of communication round $t$ Therefore, we define the state space in communication round $t$ as 
    \begin{equation}
        \mathbf{s}_t = \{ h^{n,c}_t, \sum_{i=1}^{t-1} \left( \Gamma\left(\phi^{}_{i}\left(v\right)\right) + \chi_i +  \psi_i\right)\}_{\forall n}.
    \end{equation}
    
    \item Action space $\mathcal{A}$. Since the cutting point $v$ is selected from $\mathcal{V}$ that composed of $V-1$ layers in each communication round, we define the state space in communication round $t$ as $\mathbf{a}_t = \{1, 2, ..., V-1 \}$.

    \item State transition probability $\mathcal{P}$. Let $\mathcal{P} \left(s_{t-1}|s_{t}, a_{t}\right)$ be the probability of transitioning from state $s_{t-1}$ to state $s_{t}$ under action $a_{t}$.

    \item Reward $\mathcal{R}$. Reward $r_t$ is designed to evaluate the quality of a learning policy under state-action pair $\left( \mathbf{s}_t, \mathbf{a}_t \right)$, which is defined as
    \begin{equation}\label{reward}
        r_t\left( \mathbf{s}_t, \mathbf{a}_t \right) = \begin{cases}
     \Gamma\left(\phi^{}_{t}\left(v\right)\right) + \chi_t +  \psi_t, &    log(1+\frac{\phi_t\left(v\right)}{q}) \geq \epsilon, \\
    C, & \mathrm{otherwise.}
  \end{cases}
    \end{equation}
where $C$ is a sufficiently large value used as a penalty.
\end{itemize}

Based on the tuple above, we further define the cumulative discounted reward for $t$-th communication round as
\begin{equation}
    U_t = \lim_{T\rightarrow{+\infty}}\sum_{i=t}^{T} \gamma^{i-t} r_{i}\left(\mathbf{s}_{i}, \boldsymbol{a}_{i} \right),
\end{equation}
where $\gamma \in \left( 0, 1\right]$ is the discount factor for weighting future rewards. Then the MDP problem is formulated, aiming to find an optimal policy $\pi^*$ that maximizes the expected long-term discounted rewards, i.e.,
\begin{equation}
    \pi^* = \arg \max_{\pi} \mathbb{E}_{\pi}\left[U_t\right].
\end{equation}

To tackle the MDP problem, the DDQN algorithm defines an agent that interacts with the environment to choose better actions based on a certain policy $\pi$ for maximizing long-term discounted rewards. To this end, the DDQN introduces a state-action function $Q^{\pi}\left( \mathbf{s}_t, \mathbf{a}_t; \boldsymbol{\theta} \right)$  for a certain policy $\pi$ as the expected future long-term reward for a state-action pair $\left( \mathbf{s}_t, \mathbf{a}_t \right)$, which is presented by
\begin{equation}
	Q^{\pi}\left( \mathbf{s}_{t}, \boldsymbol{a}_{t}; \boldsymbol{\theta}\right)  = \mathbb{E}_{\pi}\left[U_t|\mathbf{s}_{t}, \boldsymbol{a}_{t}\right],
\end{equation}
where $\boldsymbol{\theta} $ is the parameter vector of the Q-network.

To find the optimal policy $\pi^*$, it is equivalent to obtaining the optimal action-value function $Q^{*}\left( \mathbf{s}_{t}, \boldsymbol{a}_{t}; \boldsymbol{\theta}\right)$, which can be achieved through the Bellman equation as 
\begin{equation}\label{bellman}
    Q^{^*}\left( \mathbf{s}_{t}, \boldsymbol{a}_{t}; \boldsymbol{\theta}\right) = r_t + \gamma \max_{\boldsymbol{a}_{t+1}} Q^{^*}\left( \mathbf{s}_{t+1}, \boldsymbol{a}_{t+1}; \boldsymbol{\theta}\right).
\end{equation}

The optimal action-value function $Q^{^*}$ can be obtained by optimizing the parameter vector $\boldsymbol{\theta}$ of the Q-network. To this end, the DDQN algorithm optimizes the parameter $\boldsymbol{\theta}$ by minimizing the following loss function 
\begin{equation} \label{DDQNloss}
\begin{split}
	\mathcal{L }(\boldsymbol{\theta}) & =  \left(r_t+\gamma \max _{\boldsymbol{a}_{t+1}} Q\left(\mathbf{s}_{t+1}, \arg \max _{\boldsymbol{a}_{t+1}} Q\left(\mathbf{s}_{t+1}, \mathbf{a}_{t+1} ; \boldsymbol{\theta} \right) ; \hat{\boldsymbol{\theta}} \right)\right. \\
		& \left. -Q(\mathbf{s}_{t}, \mathbf{a}_{t}; \boldsymbol{\theta})\right)^2.
\end{split}
\end{equation}

Then, a gradient descent method is employed to minimize the loss function  $\mathcal{L }(\boldsymbol{\theta})$. As a result, the optimal policy is achieved by obtaining the optimal parameter vector $\boldsymbol{\theta}^* $.

Following the proposed optimization methods above, we now introduce a joint cutting point control and resource allocation strategy for $\mathcal{P} 1$. Specifically, we employ the DDQN to optimize the cutting point selection subproblem by reformulating $\mathcal{P} 2.2$ as a MDP.
During each exploration in the DDQN algorithm, the agent take an action to acquire rewards as defined in equation \eqref{reward} where $\chi_t$ and $\psi_t$ are obtained through resolving $\mathcal{P} 2.1$ by convex optimization technique.

The detailed procedure is shown in Algorithm \ref{ProAlg1}.

\begin{algorithm}\label{ProAlg1}
	\SetAlgoLined
	\KwIn{Initialize parameter vector of Q-networks $\boldsymbol{\theta}^{1}$; Initialize the experience buffer; Maximum episode number $L_{\max}$.}
	\For{episode $\ell =1 $ \KwTo\  $L_{\max}$  }{
		
		Reset the initial state $\mathbf{s}_{1}$;%\sim 
   
		\For{communication round $t =1 $ \KwTo\ $T$ }{
                 
			DQN agent selects discrete action $\boldsymbol{a}_{t}$ based on the observed state $\mathbf{s}_{t}$;
			
			Obtain the optimal $f^{n, c*}_t$, $f^{n, s*}_t$, $p^{n*}_t$ and $B^{n*}_t$ by resolving $\mathcal{P} 2.1$;
					
			Calculate the reward $r_{t}$ with $f^{n, c*}_t$, $f^{n, s*}_t$, $p^{n*}_t$ and $B^{n*}_t$; 
			
			Observe the next $\mathbf{s}_{t+1}$;
					
			Add transition $(\mathbf{s}_{t}, \boldsymbol{a}_{t}, r_{t}, \mathbf{s}_{t+1})$ to the replay buffer;
			
			Sample a minibatch from the replay buffer;
				
		    Update DQN network by the gradient descent method: $\boldsymbol{\theta}^{t+1} \leftarrow \boldsymbol{\theta}^{t}$;
				
		}
% 		\Until{($f(x_k)>f(x_{k-1})$)}
	}	
	\caption{The joint CCC  strategy for $\mathcal{P} 1$.}
\end{algorithm}

\begin{figure*}
	\centering
	\subfigure[MINIST]{
		\begin{minipage}[t]{0.32\linewidth}
			\centering
			\includegraphics[width=2.2in]{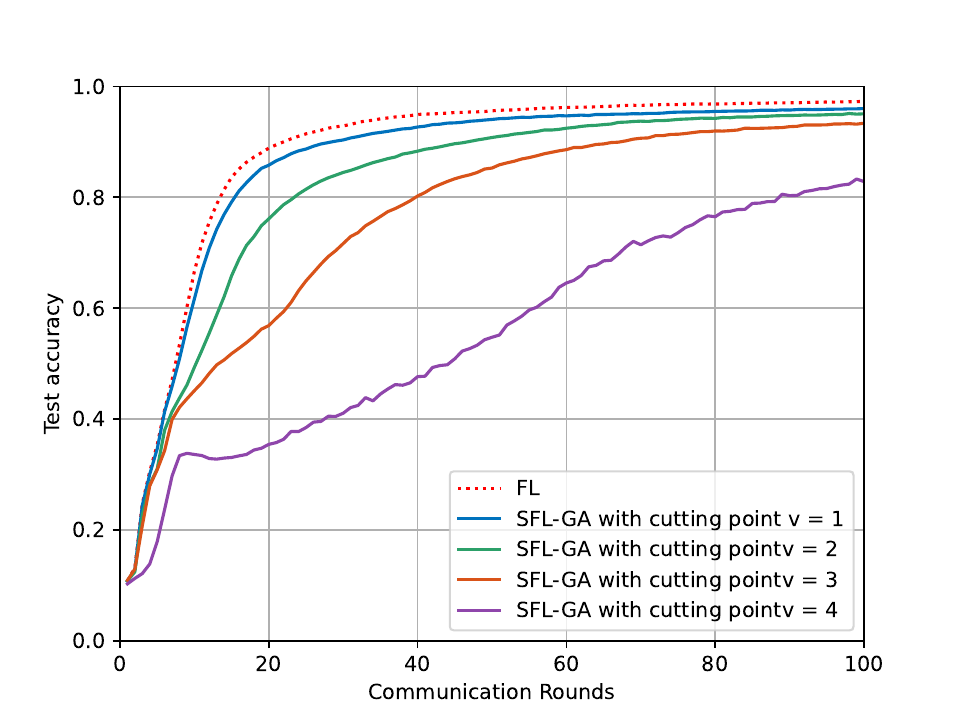}
		\end{minipage}
	}%
	\subfigure[FMINIST]{
		\begin{minipage}[t]{0.32\linewidth}
			\centering
			\includegraphics[width=2.2in]{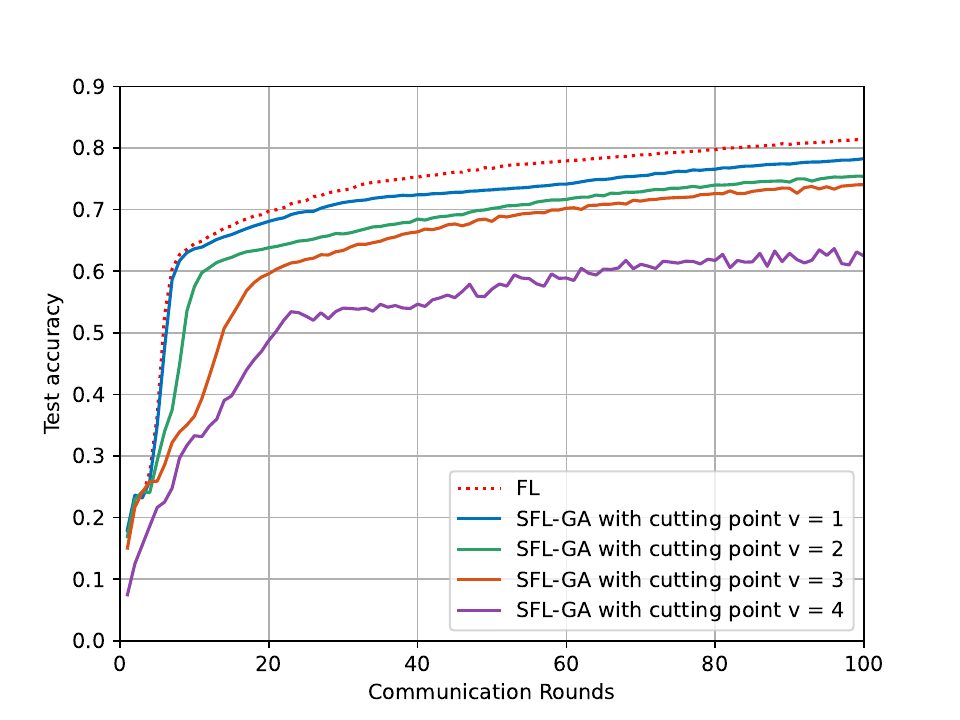}
		\end{minipage}
	}%
	%此处的空行很重要，想让图片在什么地方换行就在代码对应位置空行
	\subfigure[CIFAR-10]{
		\begin{minipage}[t]{0.32\linewidth}
			\centering
			\includegraphics[width=2.2in]{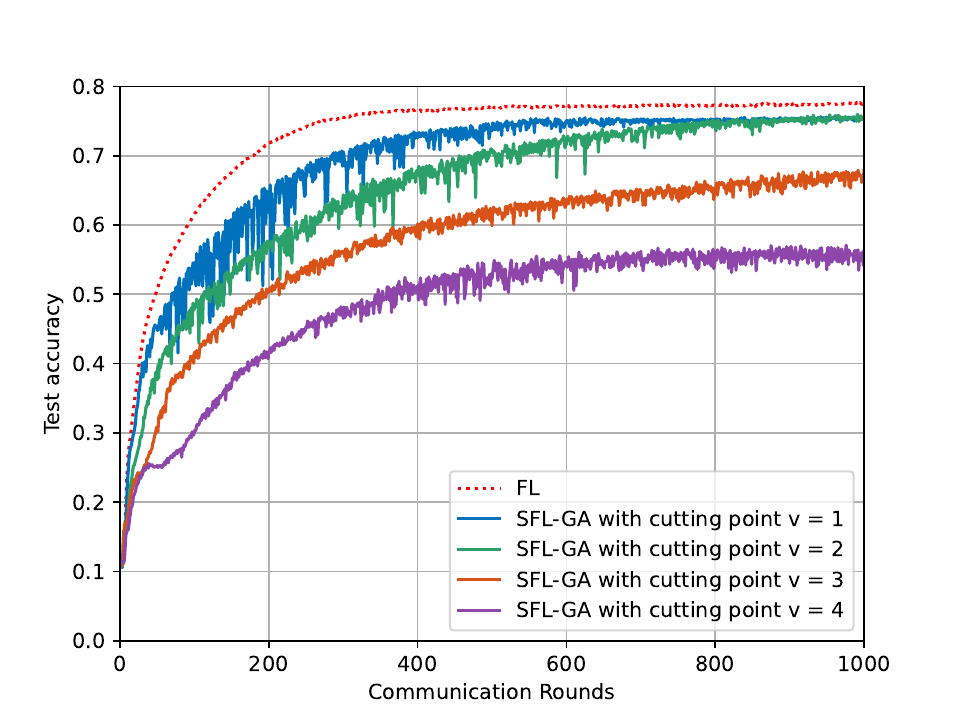}
		\end{minipage}
	}%
	\centering
	\caption{Convergence performance evaluation over different cutting layer.}
	\label{ConvergenceCutting}
\end{figure*}

\begin{figure*}
	\centering
	\subfigure[MNIST]{
		\begin{minipage}[t]{0.32\linewidth}
			\centering
			\includegraphics[width=2.2in]{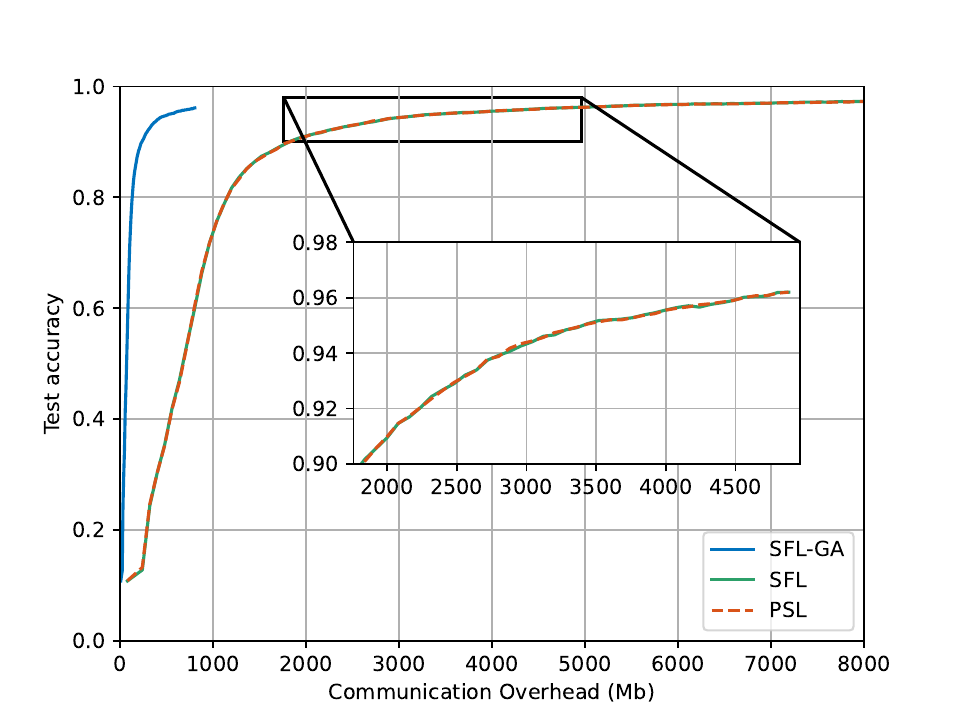}
		\end{minipage}
	}%
	\subfigure[FMNIST]{
		\begin{minipage}[t]{0.32\linewidth}
			\centering
			\includegraphics[width=2.2in]{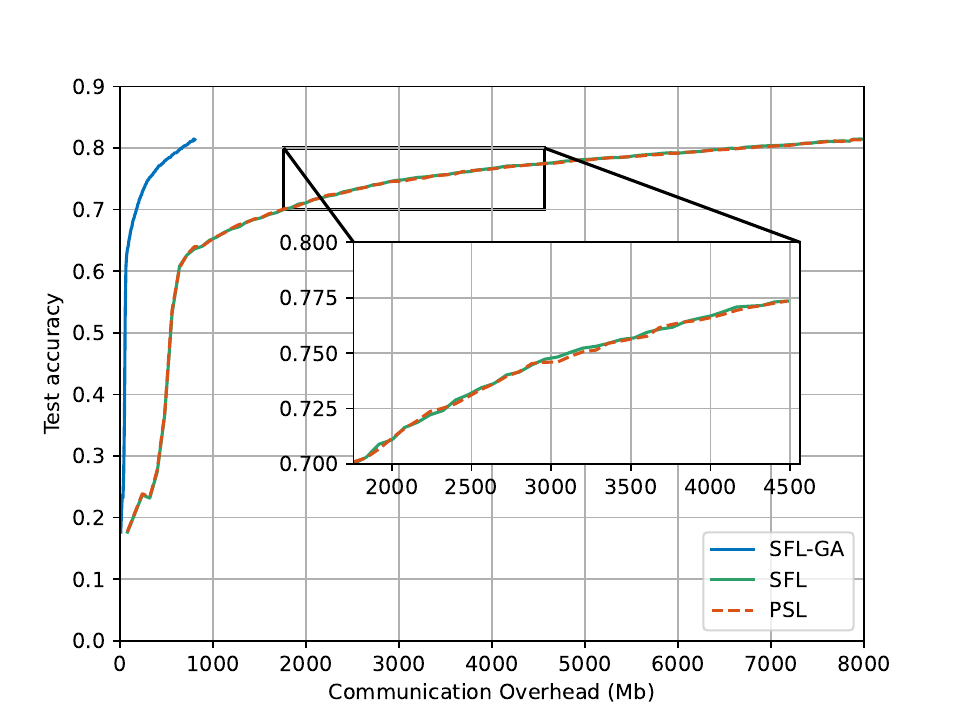}
		\end{minipage}
	}%
	%此处的空行很重要，想让图片在什么地方换行就在代码对应位置空行
	\subfigure[CIFAR-10]{
		\begin{minipage}[t]{0.32\linewidth}
			\centering
			\includegraphics[width=2.25in]{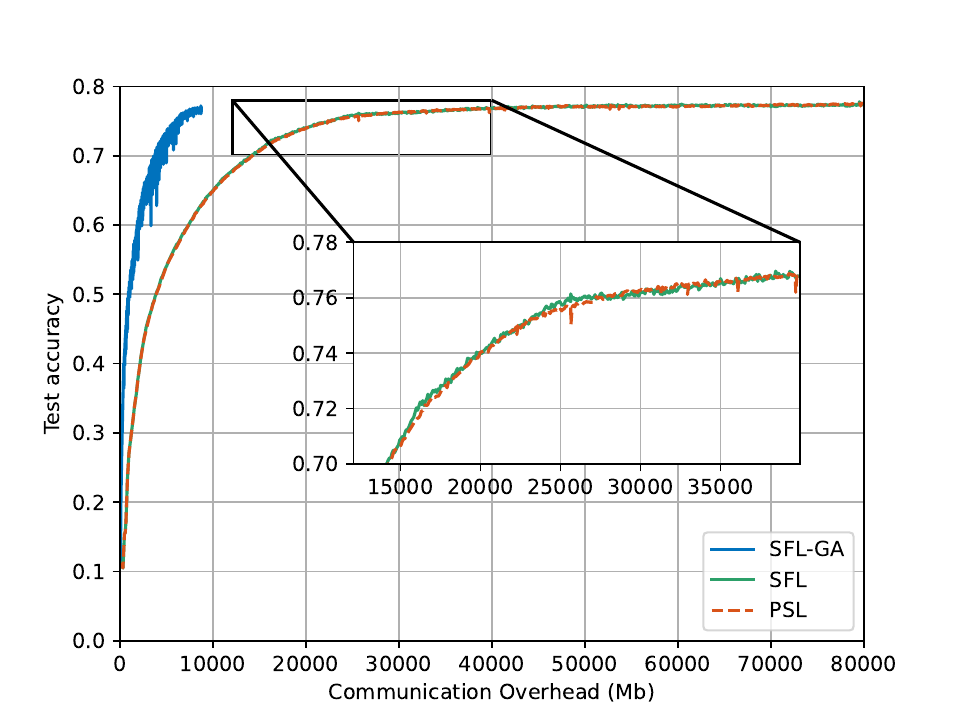}
		\end{minipage}
	}%
	\centering
	\caption{Communication overheads over different schemes.}
	\label{Communication}
\end{figure*}

\subsection{Complexity Analysis of Algorithm 1}

As described previously, the proposed Algorithm \ref{ProAlg1} integrates convex optimization method and DDQN algorithm, where the agent needs to resolve $\mathcal{P} 2.1$ before obtaining a reward. Therefore, we first analyses the complexity of solving $\mathcal{P} 2.1$. Then, the complexity of algorithm 1 is further presented.
Note that $\mathcal{P} 2.1$ is a convex optimization problem, which can be resolved with a polynomial complexity, e.g., $\mathcal{O}(N^{3.5})$ \cite{BoydCVX}.
The DDQN network is represented by an fully-connected NN (FCNN) in this work. Generally, the computational complexity of an FCNN is $ \mathcal{O}\left(\sum_k\left(2 I_k-1\right) I_{k+1}\right)$ \cite{ZhangJSAC2021}, where $k \in [0, K]$ denotes the layer index and $I_k$ represents the neuron number of hidden layers.
Let $M$ be the total number of episodes and $T$ be the number of steps per episode. Then, the overall computational complexity of Algorithm 1 is $\mathcal{O}\left( TM \left(\sum_k\left(2 I_k-1\right) I_{k+1} + N^{3.5} \right) \right)$.

\section{Simulation results}
In this section, we provide simulation results to validate the effectiveness of proposed SFL-GA and the efficiency of developed algorithm design.

\subsection{Experiment Setup}
\subsubsection{Proposed training setting}  The experiments are conducted on a environment with a server and $N = 10$ devices. The learning task is to train a Convolutional Neural Networks (CNN) model for different classification tasks. To evaluate our proposed scheme, we conduct the experiments over three different datasets: MNIST, fashion MNIST and CIFAR-10 datasets. 
We use similar model architectures as adopted in \cite{McMahanAISTATS2017} for model training. 
The max CPU-cycle frequency $f^{n,\max}_t$ for each client is $0.1$ GHz, and the total CPU-cycle frequency for server is $100$ GHz. We assume the computation workloads for each client and server are set to $\gamma^n_F = \gamma^n_B = 5.6$ MFlops and $\gamma^n_F = \gamma^n_B = 86.01 $ MFlops, respectively \cite{XuTWC2023}.
\subsubsection{Wireless communication setting} We assume that the path loss of wireless channels between devices and the edge server is given by $128.1+37.6log10(d)$ (in dB) , where $d$ represents the distance  in kilometer (km) \cite{XingTWC2018}.
 We assume that the thermal noise spectrum density $N_0 = -174$ dBm. The maximum transmit power budgets for each client and server are $p^n_{\max} = 25$ dBm and $P = 33$ dBm, respectively. The total bandwidth $B = 20$ MHz. 
 % if not specified. otherwise

\begin{figure*}
	\centering
	\subfigure[MNIST]{
		\begin{minipage}[t]{0.32\linewidth}
			\centering
			\includegraphics[width=2.2in]{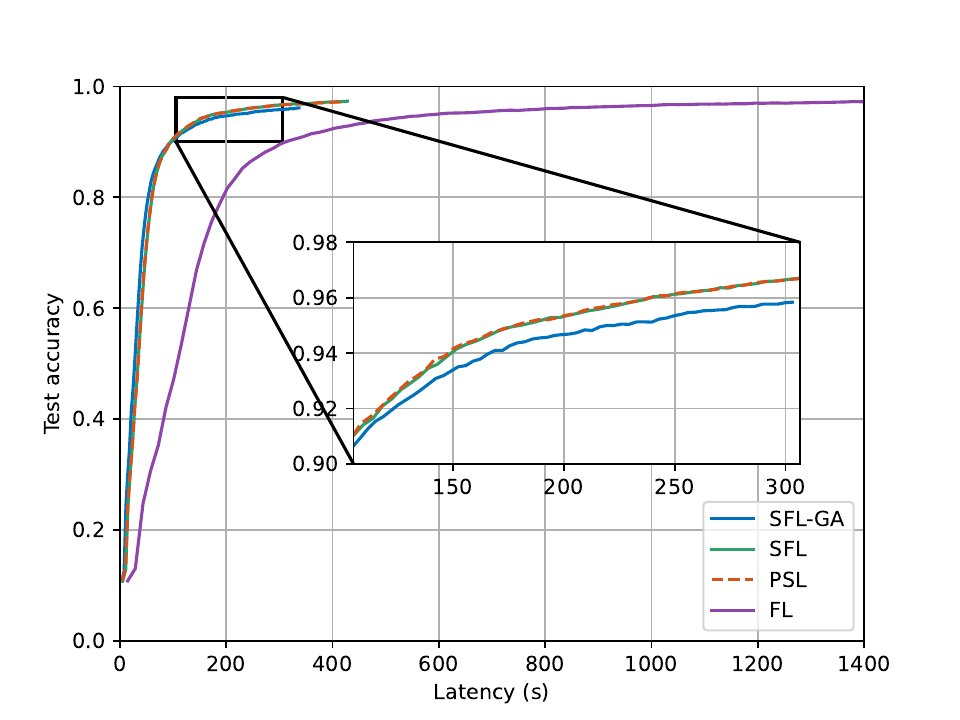}
		\end{minipage}
	}%
	\subfigure[FMNIST]{
		\begin{minipage}[t]{0.32\linewidth}
			\centering
			\includegraphics[width=2.2in]{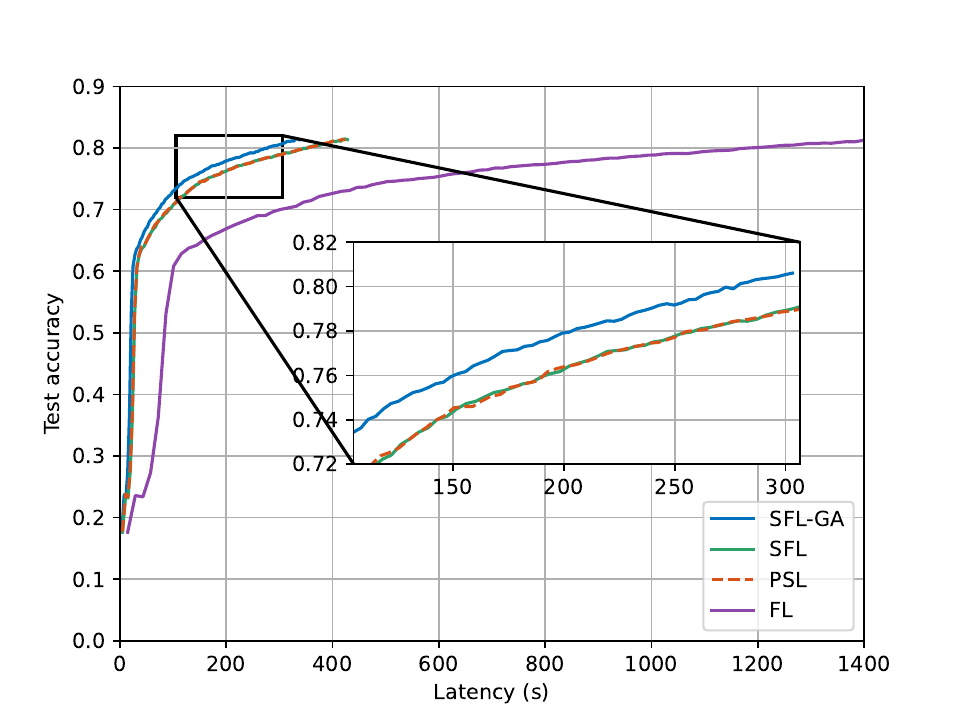}
		\end{minipage}
	}%
	%此处的空行很重要，想让图片在什么地方换行就在代码对应位置空行
	\subfigure[CIFAR-10]{
		\begin{minipage}[t]{0.32\linewidth}
			\centering
			\includegraphics[width=2.25in]{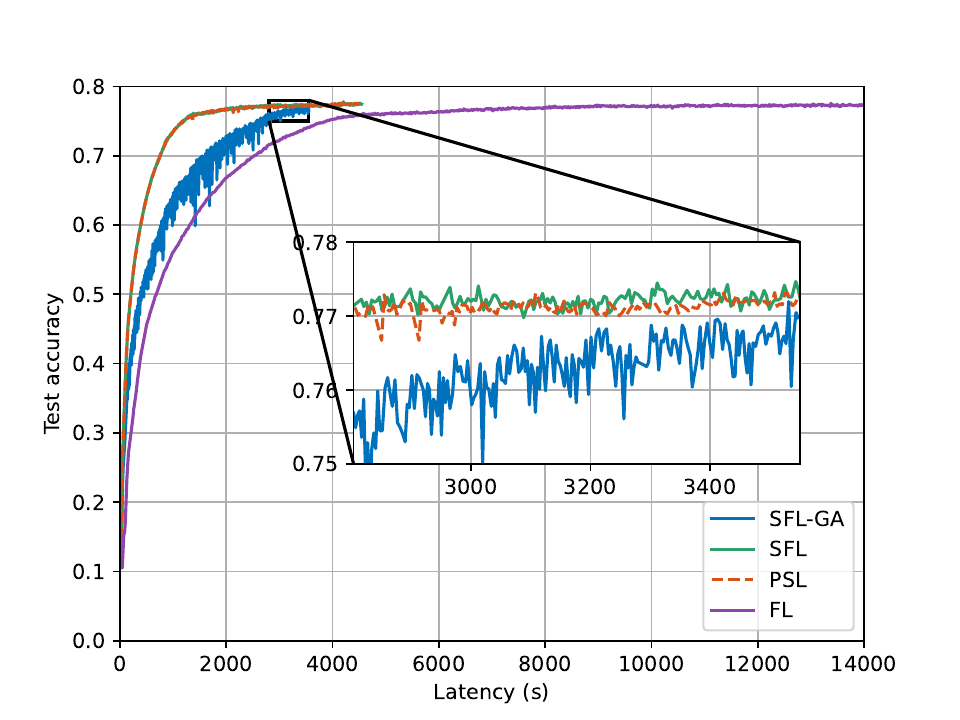}
		\end{minipage}
	}%
	\centering
	\caption{Accuracy under latency of different schemes.}
	\label{latencyScheme}
\end{figure*}

\begin{figure*}
	\centering
	\subfigure[MNIST]{
		\begin{minipage}[t]{0.32\linewidth}
			\centering
			\includegraphics[width=2.2in]{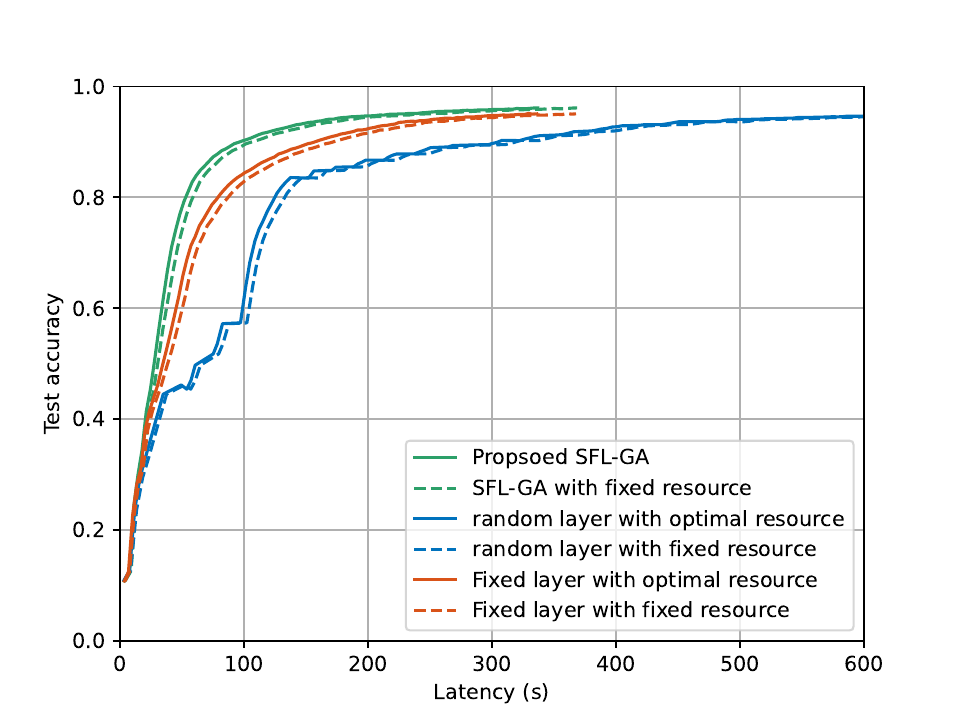}
		\end{minipage}
	}%
	\subfigure[FMNIST]{
		\begin{minipage}[t]{0.32\linewidth}
			\centering
			\includegraphics[width=2.2in]{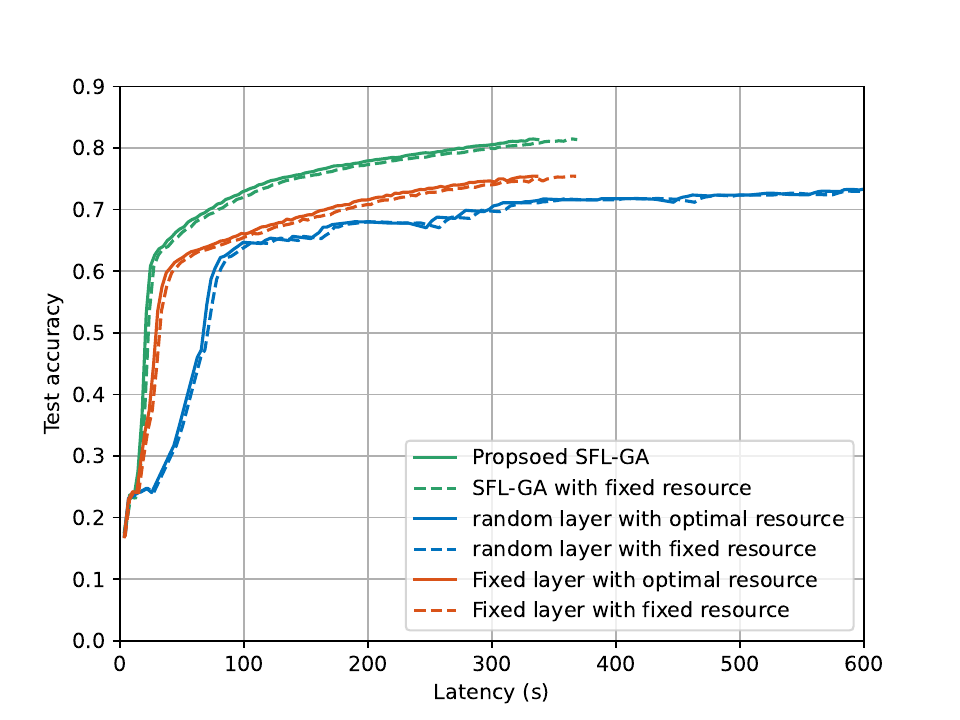}
		\end{minipage}
	}%
	%此处的空行很重要，想让图片在什么地方换行就在代码对应位置空行
	\subfigure[CIFAR-10]{
		\begin{minipage}[t]{0.32\linewidth}
			\centering
			\includegraphics[width=2.25in]{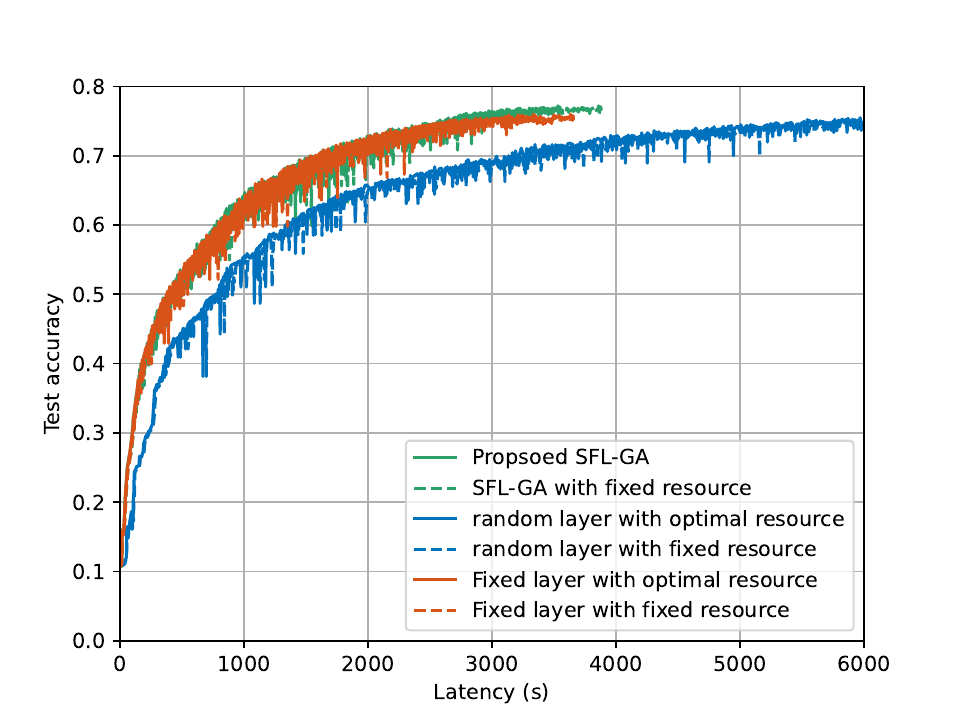}
		\end{minipage}
	}%
	\centering
	\caption{Accuracy vs latency over different resource strategies.}
	\label{latency}
\end{figure*}

\subsection{Theoretical Analyses}
Fig. \ref{ConvergenceCutting} illustrates the convergence behavior across different cutting points under various datasets. Here, the SFL scheme serves as a benchmark for validating our theoretical analyses. It is evident that SFL outperforms the proposed SFL-GA in terms of communication rounds. Moreover, the convergence performance of SFL-GA deteriorates with an increasing cutting point. For example, while SFL-GA with cutting point $v=1$ achieves approximately $95$\% test accuracy over the MINIST dataset, SFL-GA with cutting point $v=4$ only achieves about $82$\% after 100 communication rounds. 
This observation aligns with our theoretical analysis, indicating that a smaller client-side model size leads to better convergence performance for SFL-GA.
% This is consistent with our theoretical analysis that the cutting layer selection significant affects the client-side model update due to gradient aggregation, which finally determines the convergence performance of SFL-GA. 

Fig. \ref{Communication} presents communication overheads versus convergence performance across various existing schemes, including SFL-GA, traditional SFL, and PSL. It is observed that our proposed SFL-GA demonstrates greater communication efficiency compared to traditional SFL and PSL, as it achieves comparable test accuracy with significantly lower communication overhead. 
For example, the communication overhead for SFL-GA to achieve approximately $94$\% test accuracy is below 20 MB on the MNIST dataset, whereas it exceeds $40$ MB for traditional SFL. This underscores the effectiveness of our proposed SFL-GA in terms of communication efficiency. Additionally, we note that the communication overhead for PSL is marginally lower than that of traditional SFL for achieving the same test accuracy. This discrepancy arises from the fact that PSL does not necessitate client-side model aggregation, unlike SFL.

Fig. \ref{latencyScheme} illustrates the accuracy versus latency across different datasets for different schemes. 
% To demonstrate the effectiveness of our proposed SFL-GA framework, we compare it against existing benchmarks, including FL, traditional SFL, and PSL. 
As shown in Fig. \ref{latencyScheme}, FL exhibits the highest latency to achieve convergence, thereby demonstrating the poorest performance. This is due to the fact that FL updates the entire model on clients with limited computational resources. In contrast, SFL-GA, SFL, and PSL offload portions of the ML model to a server with greater computational power for model training, thus reducing the latency required for convergence. In particular, our proposed SFL-GA achieves similar test accuracy compared to SFL and PSL. However, as shown in Fig. \ref{Communication}, SFL-GA significantly reduces communication overhead, highlighting the superiority of our framework.

\subsection{Effectiveness of the Proposed Algorithms}
\begin{figure}
    \centering
    \includegraphics[width=2.5in]{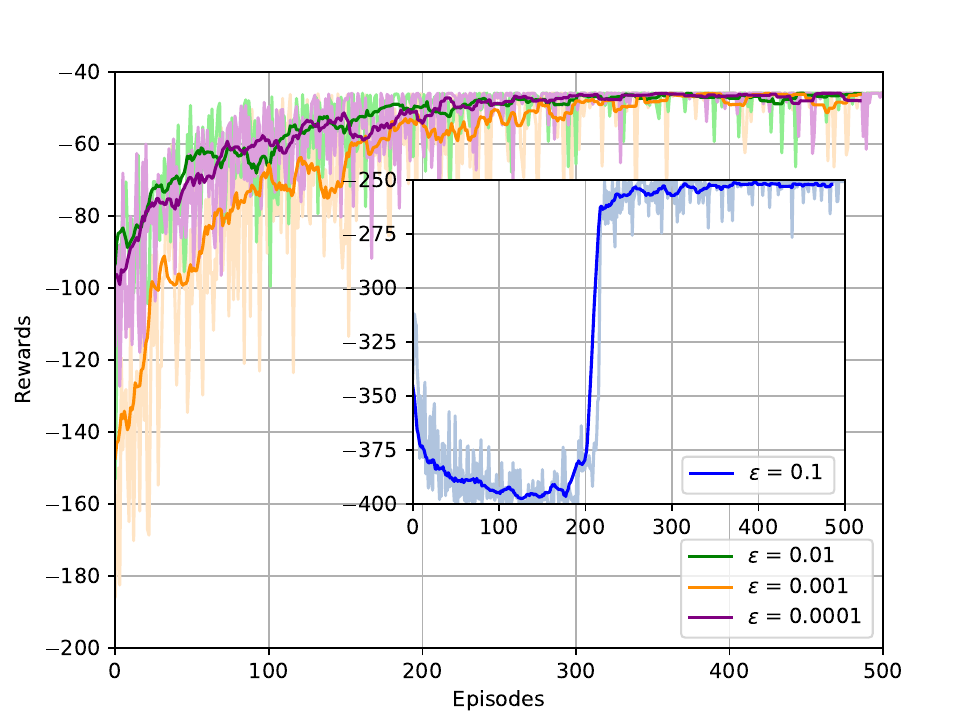}
    \caption{Convergence performance of Algorithm 1}
    \label{Algorithm1}
\end{figure}
% \begin{figure}[htbp]
% \centering
% \begin{minipage}[t]{0.48\textwidth}
% \centering
% \includegraphics[width=3.0in]{DDQN_rewards.pdf}
% \caption{Convergence performance of Algorithm 1.}
% \label{Algorithm1}
% \end{minipage}
% \begin{minipage}[t]{0.48\textwidth}
% \centering
% \includegraphics[width=2.9in]{Latency_MNIST_differentB.pdf}
% \caption{XXX.}
% \label{DifferentB}
% \end{minipage}
% \end{figure}

Fig. \ref{latency} evaluates accuracy against latency under various resource allocation strategies. We consider both fixed cutting layer and random cutting layer strategies, each assessed using optimal and fixed computation and communication resource allocations as benchmarks. It is seen that that our proposed Algorithm \ref{ProAlg1} for SFL-GA achieves the shortest latency for convergence among the benchmarks. Additionally, the selection of the cutting layer significantly impacts latency. For instance, with identical computation and communication resource allocations, Algorithm \ref{ProAlg1} substantially reduces latency compared to the random strategy across different datasets. 

Fig. \ref{Algorithm1} illustrates the convergence of the proposed Algorithm \ref{ProAlg1} under various privacy constraints. The results clearly show that the rewards converge within 500 episodes across different constraints, highlighting the effectiveness of Algorithm \ref{ProAlg1}. Meanwhile, the convergence points of the rewards vary depending on the $\epsilon$ value, underscoring the effectiveness of our designed algorithm. For instance, with $\epsilon = 0.001$, the reward converges to approximately $-45$, whereas with $\epsilon = 0.0001$, it converges to about $-250$, demonstrating the significant impact of privacy constraints on model training.

\begin{figure}
    \centering
    \includegraphics[width=2.5in]{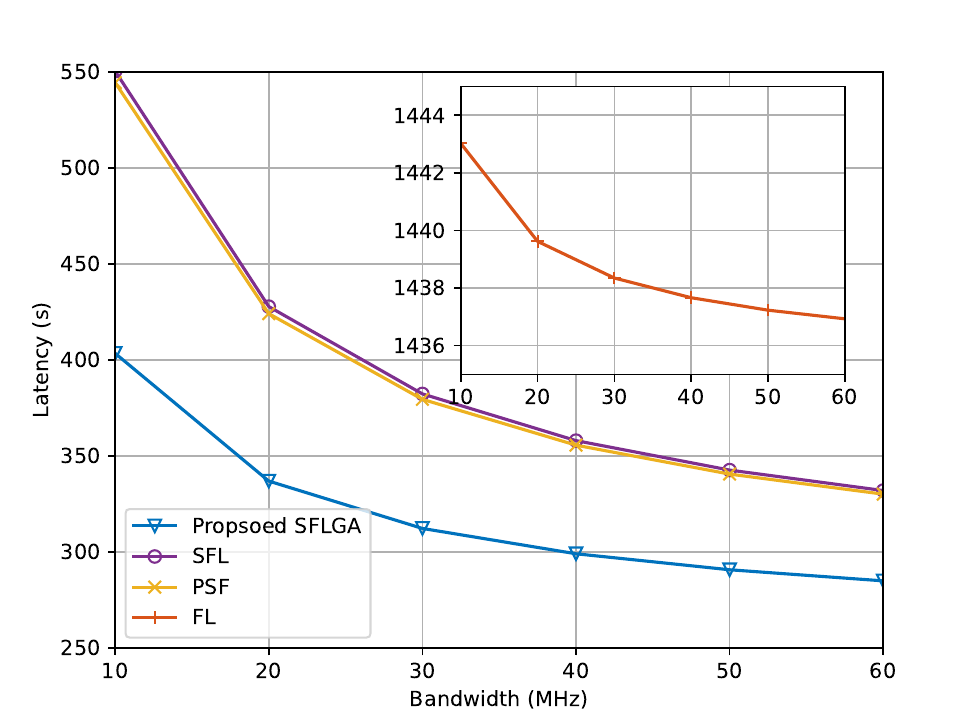}
    \caption{Latency under different bandwidth allocation.}
    \label{DifferentB}
\end{figure}

Fig. \ref{DifferentB} presents the latency under different bandwidth allocations in MNIST datasets. It is evident that latency decreases for all schemes as the available bandwidth increases. This is reasonable since more bandwidth leads to a higher transmission rate, thereby reducing communication latency. Additionally, the proposed SFL-GA achieves the lowest latency for given bandwidth budgets compared to the benchmarks, including FL, traditional SFL, and PSL. In particular, the latency of our framework is significantly lower than that of both SFL and PSL. This reduction results from the proposed gradient aggregation scheme, in which the aggregated gradients of the smashed data are broadcast to all clients. Moreover, it is worth noting that the latency of SFL is slightly higher than that of PSL. This is due to the additional communication overhead in SFL, as it requires client-side model aggregation updates in each communication round.

\section{Conclusion}
In this work, we proposed an SFL-GA framework to improve both communication and computation efficiency for traditional SFL. Specifically, the framework enabled dynamic model cutting point selection by considering wireless network environment, privacy constraints, and computational burden. Given the cutting point, the gradients of smashed data were aggregated before broadcasting, effectively reducing communication overhead. 
We theoretically analyzed the impact of the cutting point selection on convergence performance of the proposed SFL-GA framework. Furthermore, we formulated an optimization problem to further enhance the communication-and-computation efficiency of the framework, considering the model convergence rate, latency, as well as privacy leakage. To deal with the problem, an efficient joint CCC strategy was designed by integrating the DDQN algorithm and optimization method. Extensive simulation results were provided to verify the effectiveness of the proposed framework and demonstrate the efficiency of the proposed algorithm.

% if have a single appendix:
%\appendix[Proof of the Zonklar Equations]
% or
%\appendix  % for no appendix heading
% do not use \section anymore after \appendix, only \section*
% is possibly needed

% use appendices with more than one appendix
% then use \section to start each appendix
% you must declare a \section before using any
% \subsection or using \label (\appendices by itself
% starts a section numbered zero.)
%

\appendices
% \section{Proof of the First Zonklar Equation}
% Appendix one text goes here.

\section{Proof of Lemma 1} \label{AppenB}

To prove Lemma 1, we first derive the following auxiliary variables 
 Averaged Mini-batch Gradient Of SFL-GA, as 
\begin{equation}
\Bar{\mathbf{g}}^{n}_{t}
    =  \left[ 
\begin{matrix}
\Bar{\mathbf{g}}^{s,n}_{t}\\
  \Bar{\mathbf{g}}^{c}_{t}
\end{matrix}
\right]
    =  \left[ 
\begin{matrix}
 \frac{1}{\tau}\sum_{i = 1}^{\tau}\mathbf{g}^{s,n}_{t,i}\\
 \frac{1}{\tau}  \sum_{i = 1}^{\tau} \mathbf{g}^{c}_{t,i}
\end{matrix}
\right].    %\frac{1}{\tau} \sum_{i=0}^{\tau-1}  \mathbf{g}^{n}_{t,i},
\end{equation}

Meanwhile, we can expressed the Averaged Full-batch Gradient of SFL-GA as 
\begin{equation}
    \Bar{\mathbf{h}}^{n}_{t} 
    =  \left[ 
\begin{matrix}
 \Bar{\mathbf{h}}^{s,n}_{t}\\
 \Bar{\mathbf{h}}^{c}_{t}
\end{matrix}
\right]
=  \left[ 
\begin{matrix}
 \frac{1}{\tau} \sum_{i = 1}^{\tau}   \nabla_{\mathbf{w}^{s}} F\left( \mathbf{w}^{n}_{t,i}\right) \\ %\sum^{N}_{n =1} \rho^n
 \frac{1}{\tau}\sum_{i = 1}^{\tau} \nabla_{\mathbf{w}^{c}} \tilde{F}\left( \mathbf{w}^{}_{t,i}\right)
\end{matrix}
\right].   %\frac{1}{\tau} \sum_{i=0}^{\tau-1}  \nabla F(\mathbf{w}^{n}_{t,i}),
\end{equation}

Moreover, the Averaged Full-batch Gradient of SFL can be written as 
\begin{equation}
    \nabla \Bar{F}(\mathbf{w}^{n}_{t})
    =  \left[ 
\begin{matrix}
\nabla \Bar{F}(\mathbf{w}^{s,n}_{t})\\
 \nabla \Bar{F}(\mathbf{w}^{c}_{t})
\end{matrix}
\right]
=  \left[ 
\begin{matrix}
 \frac{1}{\tau} \sum_{i = 1}^{\tau}   \nabla_{\mathbf{w}^{s}} F\left( \mathbf{w}^{n}_{t,i}\right) \\ %\sum^{N}_{n =1} \rho^n
 \frac{1}{\tau}\sum_{i = 1}^{\tau} \nabla_{\mathbf{w}^{c}} F\left( \mathbf{w}^{n}_{t,i}\right)
\end{matrix}
\right].  
\end{equation}

Then, the update of the global model between two consecutive adjacent rounds is formulated as 
\begin{equation}
    \begin{split}
        \mathbf{w}_{t+1} - \mathbf{w}_{t} &= \sum_{n =1}^{N} \rho^n \left[  \mathbf{w}^n_{t+1} - \mathbf{w}^n_{t} \right]
        % &= \sum_{n =1}^{N} \rho^n \left[  \mathbf{w}^n_{t} - \mathbf{w}^n_{t} - \eta  \sum_{i = 1}^{\tau}  \mathbf{g}^{n}_{t,i} \right]  \\
        = -\eta \tau \sum_{n=1}^{N}\rho^{n} \Bar{\mathbf{g}}^{n}_{t}.
    \end{split}
\end{equation}

According to the assumption of \textbf{L-smoothness}, the improvement on the global loss can be expressed as
\begin{equation}\label{lossimprove}
\begin{split}
     &\mathbb{E} \left(F\left(\mathbf{w}_{t+1}\right) - F\left(\mathbf{w}_{t}\right)\right) %\leq\\%\mathbb{E} \left[\left\langle \nabla F(\mathbf{w}_{t}), \mathbf{w}_{t+1} - \mathbf{w}_{t}\right\rangle \right]  + \frac{L}{2} \mathbb{E} \left[ \|\mathbf{w}_{t+1} - \mathbf{w}_{t}\|^2 \right] \\ %\left(L-Smoothness\right)
     % & = -\eta \tau \left\langle \nabla F(\mathbf{w}_{t}), \sum_{n=1}^{N}\rho^{n} \Bar{\mathbf{g}}^{n}_{t}\right\rangle  + \frac{L\eta^2\tau^2 }{2} \| \sum_{n=1}^{N}\rho^{n} \Bar{\mathbf{g}}^{n}_{t} \|^2 \\
     % & = -\eta \tau \mathbb{E} \left[ \left\langle \nabla F(\mathbf{w}_{t}), \sum_{n=1}^{N}\rho^{n} \left( \Bar{\mathbf{g}}^{n}_{t} - \Bar{\mathbf{h}}^{n}_{t} \right)+ \sum_{n=1}^{N}\rho^{n} \Bar{\mathbf{h}}^{n}_{t} \right\rangle \right] \\
     % & ~~~~ + \frac{L\eta^2 \tau^2 }{2} \mathbb{E} \left[  \left\| \sum_{n=1}^{N}\rho^{n} \left( \Bar{\mathbf{g}}^{n}_{t} - \Bar{\mathbf{h}}^{n}_{t} \right)+ \sum_{n=1}^{N}\rho^{n} \Bar{\mathbf{h}}^{n}_{t} \right\|^2\right] \\
     \overset{(a)}{\leq} -\eta \tau  \left\langle \nabla F(\mathbf{w}_{t}), \sum_{n=1}^{N}\rho^{n} \Bar{\mathbf{h}}^{n}_{t} \right\rangle  \\
     & + L\eta^2 \tau^2 \mathbb{E} \left[ \left\| \sum_{n=1}^{N}\rho^{n} \left( \Bar{\mathbf{g}}^{n}_{t} - \Bar{\mathbf{h}}^{n}_{t} \right) \right\|^2  + \left\| \sum_{n=1}^{N}\rho^{n} \Bar{\mathbf{h}}^{n}_{t} \right\|^2 \right] \\ %\mathbb{E} \left[  \right]
     % & = -\eta \mathbb{E} \left[ \left\langle \nabla F(\mathbf{w}_{t}), \sum_{n=1}^{N}\rho^{n} \mathbf{h}^{n}_{t} -\nabla F(\mathbf{w}_{t}) + \nabla F(\mathbf{w}_{t}) \right\rangle \right] \\ 
     % & +  L\eta^2 \mathbb{E} \left[ \| \sum_{n=1}^{N}\rho^{n} \left( \mathbf{d}^{n}_{t} - \mathbf{h}^{n}_{t} \right) \|^2  + \| \sum_{n=1}^{N}\rho^{n} \mathbf{h}^{n}_{t}  + \nabla F(\mathbf{w}_{t})- \nabla F(\mathbf{w}_{t})\|^2 \right] \\
     % & \overset{(b)}{\leq} -\eta\| \nabla F(\mathbf{w}_{t}) \|^2 - \eta \mathbb{E} \left[ \left\langle \nabla F(\mathbf{w}_{t}), \sum_{n=1}^{N}\rho^{n} \mathbf{h}^{n}_{t} - \nabla F(\mathbf{w}_{t}) \right\rangle \right]\\
     % & +  L\eta^2 \mathbb{E} \left[ \| \sum_{n=1}^{N}\rho^{n} \left( \mathbf{d}^{n}_{t} - \mathbf{h}^{n}_{t} \right) \|^2  + 2\| \sum_{n=1}^{N}\rho^{n} \mathbf{h}^{n}_{t} - \nabla F(\mathbf{w}_{t}) \|^2  + 2\|  \nabla F(\mathbf{w}_{t})\|^2 \right] \\
     & \overset{(b)}{=} -\frac{\eta\tau}{2} \left\| \nabla F(\mathbf{w}_{t}) \right\|^2  + L\eta^2 \tau^2   \mathbb{E} \left[ \left\| \sum_{n=1}^{N}\rho^{n} \left( \Bar{\mathbf{g}}^{n}_{t} - \Bar{\mathbf{h}}^{n}_{t} \right) \right\|^2 \right] \\ 
     & + \! \! \! \frac{\eta\tau}{2} \! \! \left\| \nabla F(\mathbf{w}_{t}) \! \!  - \sum_{n=1}^{N}\rho^{n} \Bar{\mathbf{h}}^{n}_{t} \right\|^2 \! \!  + \! \!  \frac{\eta \tau  \left(2 L\eta \tau - 1 \right)  }{2}\left\| \sum_{n=1}^{N}\rho^{n} \Bar{\mathbf{h}}^{n}_{t} \right\|^2 \\
     & \overset{(c)}{\leq} - \frac{\eta\tau}{2} \left\| \nabla F(\mathbf{w}_{t}) \right\|^2 + \frac{\eta\tau}{2} \mathbb{E} \left[  \left\| \nabla F(\mathbf{w}_{t}) - \sum_{n=1}^{N}\rho^{n} \Bar{\mathbf{h}}^{n}_{t} \right\|^2 \right] \\
     & + L\eta^2 \tau^2 \mathbb{E} \left[  \left\| \sum_{n=1}^{N}\rho^{n} \left( \Bar{\mathbf{g}}^{n}_{t} - \Bar{\mathbf{h}}^{n}_{t} \right) \right\|^2 \right] \\
     & \overset{(d)}{=} - \frac{\eta\tau}{2} \left\| \nabla F(\mathbf{w}_{t}) \right\|^2 + \frac{\eta\tau}{2} \mathbb{E} \left[  \left\| \nabla F(\mathbf{w}_{t}) - \sum_{n=1}^{N}\rho^{n} \Bar{\mathbf{h}}^{n}_{t} \right\|^2 \right]\\
     & + L\eta^2 \tau^2 \sum_{n=1}^{N}\left(\rho^{n}\right)^2 \frac{1}{\tau^2} \sum_{i = 1}^{\tau}\mathbb{E} \left[  \left\| \left( \mathbf{g}^{n}_{t} - \mathbf{h}^{n}_{t} \right) \right\|^2 \right] \\
     & \overset{(e)}{\leq} - \frac{\eta\tau}{2}\left\| \nabla F(\mathbf{w}^{}_{t}) \right\|^2  + L\eta^2 \tau \sigma^2 \sum_{n=1}^{N}\left(\rho^{n}\right)^2 \\
     & + \frac{\eta\tau}{2} \sum_{n=1}^{N}\rho^{n}  \underbrace{ \mathbb{E} \left[  \| \nabla F(\mathbf{w}^{}_{t}) -  \Bar{\mathbf{h}}^{n}_{t} \|^2 \right] }_{A_1},
\end{split}
\end{equation} 
where $(a)$ results from the facts that $\mathbb{E} \left[ \Bar{\mathbf{g}}^{n}_{t} - \Bar{\mathbf{h}}^{n}_{t} \right]= \frac{1}{\tau} \sum_{i = 1}^{\tau} \mathbb{E} \left[\mathbf{g}^{n}_{t} - \mathbf{h}^{n}_{t} \right]= 0$ and $\| a +b\|^2 \leq 2\| a \|^2 + 2\| b\|^2$.  $(b)$ comes from the fact that $2\left\langle a,b\right\rangle = \|a\|^2 + \|b\|^2 -\|a-b\|^2$. $(c)$ is hold when $\eta \leq \frac{1}{2L\tau}$. 
$(d)$ comes form the fact that $\mathbb{E}\left[ \left\langle  \Bar{\mathbf{g}}^{i}_{t} - \Bar{\mathbf{h}}^{i}_{t}, \Bar{\mathbf{g}}^{j}_{t} - \Bar{\mathbf{h}}^{j}_{t} \right\rangle \right] = 0, \forall i\ne  j$. $(e)$ is achieved due to the Jensen Inequality.

To find the upper bound of \eqref{lossimprove}, we applied the inequality of $\|a +b\|^2 \leq 2\| a \|^2 + 2\| b\|^2$ on $A_1$ that
\begin{equation} \label{A_1}
    \begin{split}
        & A_1 = \mathbb{E} \left[  \left\| \left( \nabla F(\mathbf{w}^{}_{t}) - \nabla \Bar{F}(\mathbf{w}^{n}_{t}) \right) + \left( \nabla \Bar{F}(\mathbf{w}^{n}_{t})-  \Bar{\mathbf{h}}^{n}_{t} \right) \right\|^2 \right]\\  %\left\langle \right\rangle
        & \overset{}{\leq} 2  \mathbb{E} \left[   \left\| \left( \nabla F(\mathbf{w}^{}_{t}) - \nabla \Bar{F}(\mathbf{w}^{n}_{t}) \right) \right\|^2  +   \left\| \left( \nabla \Bar{F}(\mathbf{w}^{n}_{t})-  \Bar{\mathbf{h}}^{n}_{t} \right) \right\|^2  \right] \\
        % & = 2  \left\| \left( \nabla F(\mathbf{w}^{}_{t}) - \frac{1}{\tau} \sum_{i = 1}^{\tau}  \nabla F(\mathbf{w}^{n}_{t,i}) \right) \right\|^2  + 2 \left\| \frac{1}{\tau} \sum_{i = 1}^{\tau} \left( \nabla F(\mathbf{w}^{n}_{t,i})-  \mathbf{h}^{n}_{t,i} \right) \right\|^2\\
        &\overset{(e)}{\leq} \frac{2 }{\tau} \sum_{i = 1}^{\tau} \mathbb{E} \left( \left\| \nabla F(\mathbf{w}^{}_{t}) \! \! \! -  \nabla F(\mathbf{w}^{n}_{t,i})  \right\|^2 \! \! \! +\left\| \nabla F(\mathbf{w}^{n}_{t,i})\!  - \!   \mathbf{h}^{n}_{t,i} \right\|^2 \right)\\
        & \overset{(f)}{\leq} \frac{2}{\tau} \sum_{i = 1}^{\tau} \mathbb{E} \left\| \left( \nabla F(\mathbf{w}^{}_{t}) - \nabla F(\mathbf{w}^{n}_{t,i}) \right) \right\|^2  +  2\Gamma\left(\phi^{}_{t}\left(v\right)\right)\\
        & \overset{(g)}{\leq} \frac{2L^2}{\tau} \sum_{i = 1}^{\tau}  \underbrace{ \mathbb{E} \left\|  \mathbf{w}^{}_{t} - \Bar{\mathbf{w}}^{n}_{t,i} \right\|^2}_{A_2}  +  2\Gamma\left(\phi^{}_{t}\left(v\right)\right),
        % &\overset{(f)}{\leq}
    \end{split}
\end{equation}
where $(e)$ comes form the Jensen's inequality, and $(f)$ is achieved due to \textbf{Assumption 4}. $(g)$ follows the Lipschitz-smooth property. $\Bar{\mathbf{w}}^{n}_{t,i}$ is the model of client $n$ obtained after the $i$-th local update based on the SFL gradient $\nabla F(\mathbf{w}^{n}_{t,i}; \xi)$. % where (e) results from the fact that $\| a+b\|^2 \leq  2\|a\|^2 + 2\|b\|^2$. 
Similar with \eqref{A_1}, we further find the upper bound of $A_2$ in \eqref{A_1} as 
\begin{equation}\label{A_2}
    \begin{split}
        &A_2 = \eta^2\mathbb{E}  \left\| \sum_{j = 1}^{i} \nabla F\left( \mathbf{w}^{n}_{t,i}; \xi \right)\right\|^2\\ 
        % & A_2 \leq \\ 
        &\leq 2 \eta^2 \mathbb{E} \left\| \sum_{j = 1}^{i} \left[\nabla F\left( \mathbf{w}^{n}_{0,i}; \xi \right) - \nabla F\left( \mathbf{w}^{n}_{t,i}\right)\right] \right\|^2 \\
        & ~~~+ 2\eta^2 \mathbb{E}  \left\| \sum_{j = 1}^{i} \nabla F\left( \mathbf{w}^{n}_{t,i}\right) \right\|^2 \\
        &  \overset{(h)}{\leq} 2 i \eta^2 \mathbb{E} \left[  \sum_{j = 1}^{i}  \left\| \nabla F\left( \mathbf{w}^{n}_{t,j}; \xi \right) \! \! \! - \! \!  \nabla F\left( \mathbf{w}^{n}_{t,j} \right) \right\|^2 \! + \! \! \left\|\nabla F\left( \mathbf{w}^{n}_{t,j} \right) \right\|^2 \right] \\
        % & \overset{(i)}\leq 2 i \eta^2 \sum_{j = 1}^{i}  \sigma^2 +  2\eta^2 i \sum_{j = 1}^{i} \left\|\nabla F\left( \mathbf{w}^{n}_{t,j} \right) \right\|^2 \\
        & \overset{(i)}{\leq} 2 i \eta^2 \sum_{j = 1}^{i}  \sigma^2 +  2\eta^2 i \sum_{j = 1}^{\tau} \mathbb{E} \left\|\nabla F\left( \mathbf{w}^{n}_{t,j} \right) \right\|^2,
    \end{split}
\end{equation}
where $(h)$ results from Cauchy–Schwarz inequality. $(i)$ is hold from \textbf{Assumption 2}. Applying $\sum_{i = 1}^{\tau_{}} i = \frac{\tau_{}\left(\tau_{} - 1 \right)}{2}$, we can obtain 
 \begin{equation}\label{eq49}
	\begin{split}
		& \sum_{i = 1}^{\tau_{}} \mathbb{E}   \left\|  \mathbf{w}^{n}_{t} \! - \! \mathbf{w}^{n}_{t,i} \right\|^2 \! \leq \! \eta^2 \tau_{} \! \left( \tau_{} -1\right) \left(\sigma^2  \! + \! \sum_{j = 1}^{\tau_{} } \mathbb{E} \left\|\nabla F\left( \mathbf{w}^{n}_{t,j} \right) \right\|^2\right)  \\
		& \leq \eta^2 \tau_{}\left( \tau_{} -1\right) \! \! \!  \left(\! \! \sigma^2 \! \! \!  + 2 L^2 \sum_{j = 1}^{\tau_{} } \mathbb{E}  \left\| \Bar{\mathbf{w}}^{n}_{t,j}  \! \!  - \mathbf{w}_{t} \right\|^2 \! \! \!  +  \!  2  \sum_{j = 1}^{\tau_{} } \left\|  \nabla F(\mathbf{w}_{t}) \right\|^2 \right)\\
            & = \eta^2 \tau_{}\left( \tau_{} -1\right)  \sigma^2 + 2\eta^2 \tau^(2)_{}\left( \tau_{} -1\right) \left\|  \nabla F(\mathbf{w}_{t}) \right\|^2 \\
            & ~~~ + 2L^2\eta^2 \tau_{}\left( \tau_{} -1\right) \sum_{j = 1}^{\tau_{} } \mathbb{E} \left\| \Bar{\mathbf{w}}^{n}_{t,j}  - \mathbf{w}_{t} \right\|^2. 
	\end{split}
\end{equation}
By rearranging \eqref{eq49}, we have
\begin{equation}\label{D_1_1}
	\begin{split}
		\sum_{i = 1}^{\tau_{}}  \left\|  \mathbf{w}_{t} -  \Bar{\mathbf{w}}_{t,i} \right\|^2 & \leq \frac{\eta^2 \sigma^2 \tau_{}\left( \tau_{} -1\right)}{1 - 2 L^2 \eta^2 \tau_{}\left( \tau_{} -1\right)} \\
        & + \frac{ 2 \eta^2  \tau^{2}_{}\left( \tau_{} -1\right)}{1 - 2 L^2 \eta^2 \tau_{}\left( \tau_{} -1\right)}\left\|  \nabla F(\mathbf{w}_{t}) \right\|^2\\
        & = \frac{\eta^2 \sigma^2 \tau_{}\left( \tau_{} -1\right)}{1 - A} + \frac{A}{1 - A}\left\|  \nabla F(\mathbf{w}_{t}) \right\|^2
	\end{split}
\end{equation}
where $A = 2 L^2\eta^2 \tau \left( \tau -1 \right) $.

% we have $A \leq 2 L^2\eta^2 \tau^2 \leq \frac{1}{2} $ due to the fact that $\eta \leq \frac{1}{2L\tau}$. Therefore,
As a result,  $A_1$ in \eqref{A_1} is bounded by
\begin{equation}\label{A_1_1}
    \begin{split}
        A_1 & \leq  \frac{2 L^2\eta^2 \sigma^2 \left( \tau_{} -1\right)}{1-A}  + \frac{2A}{1-A}\left\|  \nabla F(\mathbf{w}_{t}) \right\|^2+  2\Gamma\left(\phi^{}_{t}\left(v\right)\right)\\
        & \overset{(j)}{\leq} \frac{5}{2} L^2\eta^2 \sigma^2 \left( \tau_{} -1\right)  +  \frac{1}{2}\left\|  \nabla F(\mathbf{w}_{t}) \right\|^2 +  2\Gamma\left(\phi^{}_{t}\left(v\right)\right),
        % & \overset{(e)}{\leq} 2 L^2\eta^2 \sigma^2 \left( \tau_{} -1\right) + \frac{1}{2}\left\|  \nabla F(\mathbf{w}_{t}) \right\|^2  +  2\Gamma\left(\phi^{}_{t}\left(v\right)\right)
    \end{split}
\end{equation}
% \begin{equation}\label{A_1_1}
%     \begin{split}
%         A_1 & \leq   \frac{2 L^2\eta^2 \sigma^2 \left( \tau_{} -1\right)}{1 - 2 L^2 \eta^2 \tau_{}\left( \tau_{} -1\right)}   +  2\Gamma\left(\phi^{}_{t}\left(v\right)\right)\\
%         & + \frac{ 4 L^2 \eta^2  \tau^{}_{}\left( \tau_{} -1\right)}{1 - 2 L^2 \eta^2 \tau_{}\left( \tau_{} -1\right)}\left\|  \nabla F(\mathbf{w}_{t}) \right\|^2\\
%         & =  \frac{2 L^2\eta^2 \sigma^2 \left( \tau_{} -1\right)}{1-A}  + \frac{2A}{1-A}\left\|  \nabla F(\mathbf{w}_{t}) \right\|^2+  2\Gamma\left(\phi^{}_{t}\left(v\right)\right)\\
%         & \overset{(j)}{\leq} \frac{5}{2} L^2\eta^2 \sigma^2 \left( \tau_{} -1\right)  +  \frac{1}{2}\left\|  \nabla F(\mathbf{w}_{t}) \right\|^2 +  2\Gamma\left(\phi^{}_{t}\left(v\right)\right),
%         % & \overset{(e)}{\leq} 2 L^2\eta^2 \sigma^2 \left( \tau_{} -1\right) + \frac{1}{2}\left\|  \nabla F(\mathbf{w}_{t}) \right\|^2  +  2\Gamma\left(\phi^{}_{t}\left(v\right)\right)
%     \end{split}
% \end{equation}
where $(j)$ results from the fact that $A \leq \frac{1}{5}$.

Plug \eqref{A_1_1} back into \eqref{lossimprove}, we have
\begin{equation} \label{lossimprove1}
    \begin{split}
        &\mathbb{E} \left( F\left(\mathbf{w}_{t+1}\right) - F\left(\mathbf{w}_{t}\right)\right) \leq \! \! \! - \frac{\eta\tau}{2}\left\| \nabla F(\mathbf{w}^{}_{t}) \right\|^2 \! \! \!  + \! \! \! L\eta^2 \tau \sigma^2 \sum_{n=1}^{N}\left(\rho^{n}\right)^2\\
        &\! \! \!  + \frac{\eta\tau}{2} \sum_{n=1}^{N}\rho^{n}  \left( \frac{5}{2} L^2\eta^2 \sigma^2 \left( \tau_{} -1\right) \! \! \!  +  \frac{1}{2}\left\|  \nabla F(\mathbf{w}_{t}) \right\|^2 \! \! \! +  2\Gamma\left(\phi^{}_{t}\left(v\right)\right) \right)\\
        & = - \frac{\eta\tau}{4} \left\| \nabla F(\mathbf{w}^{}_{t}) \right\|^2 + \frac{5 L^2\eta^3 \sigma^2 \tau_{} \left( \tau_{} -1\right)}{4} \\
        & + L\eta^2 \tau \sigma^2 \sum_{n=1}^{N}\left(\rho^{n}\right)^2  + \eta\tau \Gamma\left(\phi^{}_{t}\left(v\right)\right).
        % & = \frac{\eta\sigma^2}{2} \left( \frac{A}{1-A}  \right) + \eta\tau \frac{2A}{1-A}\left\|  \nabla F(\mathbf{w}_{t}) \right\|^2+ \eta\tau \Gamma\left(\phi^{}_{t}\left(v\right)\right) 
    \end{split}
\end{equation}
This completes the proof.

\ifCLASSOPTIONcaptionsoff
  \newpage
\fi

\end{document}